\documentclass[12pt]{article}

\usepackage{setspace}
\usepackage[authoryear]{natbib}
\usepackage{graphicx}
\usepackage{lscape}
\usepackage{pdflscape}
\usepackage{afterpage}
\usepackage{pdfpages}
\usepackage{float}
\usepackage{booktabs}
\usepackage[utf8]{inputenc}
\usepackage{indentfirst}
\usepackage{xcolor}
\usepackage{tikz}
\usepackage{pgffor}
\usepackage{dsfont}
\usepackage{amsfonts}
\usepackage{amsmath}
\usepackage{amssymb}
\usepackage{url}
\usepackage{subfig}
\usepackage{array}
\usepackage{threeparttable}
\usepackage{bbm}
\usepackage[hidelinks]{hyperref}
\usepackage{amsthm}
\usepackage{needspace}
\newcolumntype{C}[1]{>{\centering\arraybackslash}p{#1}}
\newtheorem{proposition}{Proposition}[section]

\newtheorem{remark}{Remark}

\begin{document}

\title{
\large ON THE USE OF DESIGN-BASED SIMULATIONS
\footnote{I would like to thank Luis Alvarez and seminar participants at the University of Sao Paulo and at the Rio-Sao Paulo Econometrics Workshop for excellent comments and suggestions.  I gratefully acknowledge financial support from FAPESP and CNPq. An earlier version of part of the analysis in this paper also appears in \cite{Ferman_assessment}. The present paper expands this analysis, develops additional results, and supersedes its earlier treatment in \cite{Ferman_assessment}.
}
}

\author{
Bruno Ferman\footnote{email: bruno.ferman@fgv.br; Sao Paulo School of Economics -- FGV}
}

\date{}
\maketitle

\begin{abstract}
\begingroup\onehalfspacing
\begingroup

Design-based simulations hold outcomes fixed and generate variation by resampling treatment assignment or shocks. They can be used in both methodological and applied work to assess inference procedures. However, their informativeness for these uses is not straightforward. We show that, when treatment effects are nonzero, these simulations can be misleading about spatial or serial correlation in clustered experiments, difference-in-differences, and shift-share designs. We analyze alternatives and formalize how these simulations can reveal features of the error distribution despite holding outcomes fixed. Finally, we consider  limitations in the use of design-based simulations to select among different inference methods in applied work.
\par\endgroup
\endgroup
\end{abstract}



\newpage

\section{Introduction}

Simulation-based analyses play a central role in econometric research. They are routinely used to study the finite-sample properties of estimators and inference procedures, and to illustrate potential failures of commonly used methods. An important class of such exercises consists of what we refer to as \emph{design-based simulations}, in which potential outcomes are treated as fixed and uncertainty in the simulations is generated by resampling treatment assignment, shocks, or other components of the research design.

Design-based simulations have been used in several influential methodological papers to study the performance of inference procedures. For example, \citet{Bertrand04howmuch} simulate placebo laws using state-level data on female wages from the Current Population Survey to illustrate severe size distortions of conventional inference methods in difference-in-differences (DiD) designs caused by serial correlation. Other influential examples include \citet{ConleyKelly} in studies of historical persistence, \citet{Adao} and \citet{Borusyak} in shift-share designs, and \citet{Chaisemartin_cluster} in small-strata randomized experiments.

Beyond their use as illustrative tools in methodological work, applied researchers may also use design-based simulations tailored to their empirical design as application-specific diagnostic tools. For example, \citet{Mackinnon} suggests the use of design-based simulations to assess the accuracy of asymptotic approximations, while \citet{ConleyKelly} use placebo simulations to guide the choice of the number of clusters for spatial inference and to evaluate the performance of heteroskedasticity-robust inference in their applications.\footnote{Related work uses simulation-based or other data-driven procedures tailored to specific empirical settings to assess, compare, select, or tune estimation and inference methods \citep{Ferman_assessment,HuberLechnerWunsch2013,BussoDiNardoMcCrary2014,AdvaniKitagawaSloczynski2019,AtheyImbensMetzgerMunro2024,ConleyGoncalvesHansen2018,CaoHansenKozburVillacorta2025,ConleyKellyKozbur2026}. These papers, however, generally do not focus on design-based simulations of the type studied here.}

In design-based simulations, researchers usually define a data-generating process (DGP) based on an empirical application or a specific dataset, in which potential outcomes are held fixed and the null hypothesis is satisfied by construction. Interpreting the results from such simulations, however, can be subtle. Whether the resulting rejection frequencies are informative about the validity of an inference method in the original empirical setting depends critically on how the DGP used in the design-based simulations relates to the true DGP along dimensions relevant for inference. Their informativeness therefore depends on which features of the original DGP are preserved in the simulations. Moreover, design-based simulations have often been used to illustrate issues related to the distribution of the errors (for example, serial or spatial correlation), which at first glance appears at odds with the fact that these simulations hold potential outcomes, and therefore errors, fixed.

This paper studies the properties of design-based simulations, with a special focus on their use to assess inference problems due to error dependence. We consider a simple setting in which treatment is assigned at the group level, which describes a cluster-randomized experiment or a stylized shift-share design. In the appendix, we obtain analogous results for a DiD model with serially correlated errors.  We show that, when there is a true treatment effect, simulations holding observed outcomes fixed may confound this effect with error dependence, overstating the distortions of inference methods that do not account for such dependence. While \citet[footnote~14 and Section~5.2]{ConleyKelly} note that their placebo simulations  may exhibit more spatial dependence in residuals than under the true DGP when the true coefficient is nonzero, we characterize this issue by formally showing how treatment effects affect the rejection probabilities obtained from these simulations. We also analyze two alternatives and show conditions under which they reliably distinguish settings in which error dependence distorts inference from those in which it does not: removing the estimated treatment effect before constructing the simulations, and using outcomes from periods before treatment. The analysis clarifies how these simulations, despite holding errors fixed, can be informative about characteristics of the error distribution.

We then consider the use of these simulations by applied researchers to compare inference methods in their specific application. On the one hand, we show conditions under which conditioning on a favorable residual-based assessment does not introduce bias or affect the distribution of  the treatment-effect estimator, properties that do not generally extend to standard design-based simulations that hold observed outcomes fixed. On the other hand, we show important limitations of the use of design-based simulations to assess whether error dependence is a problem. In particular, the assessment may fail to detect relevant error dependence with a high probability precisely in settings where this information would be most useful. When there are few independent groups or shocks, inference methods that allow for dependence may have poor finite-sample properties, making it particularly valuable to know whether a simpler inference procedure is reliable. However, the same limited amount of independent variation can make the assessment less informative. We discuss ways to evaluate this limitation and to improve the information available, including the use of additional pre-treatment outcomes.

Finally, we illustrate these issues in three shift-share applications \citep{Autor,Acemoglu,Dix}. In one application, the assessments provide strong evidence in favor of inference procedures that accommodate the dependence generated by common exposure to shocks. In the other two applications, the assessments indicate relevant finite-sample distortions in procedures designed to accommodate this dependence, while favoring geographically clustered inference. However, auxiliary simulations using the same share structures show that assessments of geographically clustered inference can have a substantial probability of failing to detect relevant spatial dependence beyond the chosen geographical clusters. Thus, favorable assessments do not provide strong assurance that geographically clustered inference is reliable. These applications illustrate an intrinsic limitation in the use of design-based simulations to select among alternative inference methods: precisely in settings where it would be most useful to know whether geographically clustered inference is reliable, the assessments may be less informative about this question.

We start in Section~\ref{Sec: DB in RCT} with a simple randomized experiment to illustrate the main limitations of design-based simulations and to show how they can be informative about inference procedures even when potential outcomes are stochastic. Section~\ref{Sec: DB in SS} then considers the group-assignment model and studies in detail the use of design-based simulations to assess whether error dependence leads to relevant distortions in inference. Section~\ref{sec:selection} discusses the use of design-based assessments to choose inference methods. Sections~\ref{Sec: simulations} and~\ref{Sec: SS illustrations} provide Monte Carlo and empirical illustrations, respectively. Section~\ref{sec:conclusion} concludes.

\section{Design-based simulations in experiments}
\label{Sec: DB in RCT}

We consider first a simple setting of randomized experiments, in which we can illustrate the fundamental problem that generally prevents design-based simulations from recovering the true DGP in an empirical application, even when the distribution of treatment allocation is known. We consider two frameworks. In the first, potential outcomes are fixed and uncertainty comes only from treatment assignment. In the second, potential outcomes are drawn from a population, so both the sample and treatment assignment are stochastic.

\subsection{A finite-population framework}
\label{sec:rct-finite}

We observe $i=1,\ldots,N$ individuals whose potential outcomes are given by $Y_i(0)$ and $Y_i(1)$. We treat these potential outcomes as fixed, so the only source of uncertainty is the treatment allocation $\mathbf T=(T_1,\ldots,T_N)$ \citep{Finite_pop,NBERw24003,fisher,Neyman1990}. Exactly $N_1$ individuals receive treatment, with all allocations satisfying this restriction being equally likely. The target parameter in this framework is the sample average treatment effect (SATE), $\tau_N=N^{-1}\sum_i[Y_i(1)-Y_i(0)]$. 

Given the treatment allocation, for each individual $i$ we observe  $Y_i = T_iY_i(1) + (1-T_i)Y_i(0)$.
The difference in means estimator, $\widehat\tau$, is unbiased for $\tau_N$. Under the usual finite-population regularity conditions, inference using robust standard errors is asymptotically valid, although it may be conservative when treatment effects are heterogeneous. Still, an econometrician or an applied researcher may want to study the finite-sample properties of this inference method in a specific application.

In this setting, the true DGP is characterized by $\{Y_i(0),Y_i(1)\}_{i=1}^N$ and the distribution of $\mathbf T$. Therefore, if we knew both potential outcomes for every individual, knowledge of the distribution of $\mathbf T$ would provide all the information needed to calculate the rejection probability of a test based on robust standard errors. More specifically, we could compute the $t$-statistic for each allocation and calculate the probability that its absolute value exceeds a given critical value. To assess size, each statistic would be centered at the SATE of these fixed potential outcomes, which need not be zero. Equivalently, we could calculate the coverage of a confidence interval for that SATE. However, we do not observe both potential outcomes for any individual, so this calculation cannot be carried out using the observed data alone.

A commonly used way to construct design-based simulations in this setting is to hold the vector of realized outcomes, $\mathbf Y$, fixed and consider different allocations from the known distribution of $\mathbf T$. For each allocation, we estimate the treatment effect and test the null of zero effect using the inference method being analyzed. We refer to these exercises as \emph{$\mathbf Y$-fixed design-based simulations}. They implicitly consider potential outcomes $\widetilde Y_i(0)=\widetilde Y_i(1)=Y_i$. 
The SATE is therefore zero in this artificial experiment. If we want to assess the validity of, say, inference based on robust standard errors in this setting, then we can enumerate all possible assignments and calculate the exact rejection probability of this procedure under these potential outcomes.

Let $\Gamma_{\mathrm{rob}}^Y(\mathbf Y)$ denote this rejection probability when we use a nominal $\alpha$ test with robust standard errors. Its construction depends on the observed outcome vector, the assignment distribution, and the test being evaluated. The subscript identifies the inference method, and the superscript identifies how the simulation outcomes are constructed (in this case, setting $ \widetilde Y_i(0)=\widetilde Y_i(1)=Y_i$). In practice, the researcher can approximate this probability by the rejection frequency over $B$ simulated assignments, denoted by $\widehat\Gamma_{\mathrm{rob},B}^Y(\mathbf Y)$. Since this simulation error can be made arbitrarily small by increasing $B$, we focus on the exact probability $\Gamma_{\mathrm{rob}}^Y(\mathbf Y)$.

If $Y_i(1)=Y_i(0)$ for every individual, the DGP used in these design-based simulations coincides with the original one, so these simulations recover the true rejection probability or confidence-interval coverage of the inference method being assessed. However, the fundamental challenge is that the true potential outcomes generally differ from $\{\widetilde Y_i(0),\widetilde Y_i(1)\}_{i=1}^N$. Hence, even if we know the distribution of treatment assignment, a $\mathbf Y$-fixed design-based simulation need not recover the true finite-sample rejection probability or the coverage of a confidence interval in the original application. For example, suppose the empirical distributions of $\{Y_i(0)\}_{i=1}^N$ and $\{Y_i(1)\}_{i=1}^N$ are symmetric and single-peaked, with $Y_i(1)=Y_i(0)+\tau$ for all $i$ and some $\tau>0$. When the effect is sufficiently large, the observed outcomes may combine these distributions into a double-peaked distribution. The simulation then evaluates the procedure using potential outcomes with a different distributional shape, which can affect its finite-sample performance. In Appendix~\ref{app:individual-coverage}, we provide a simple example in which the coverage reported by these simulations varies with the treatment effect, even though the true coverage is unchanged.

\begin{remark}[Relation to randomization inference]
While similar at first glance, this type of design-based simulation is conceptually distinct from randomization inference in experiments. In randomization inference, permutations of the treatment assignment are used to conduct exact finite-sample tests of a \emph{sharp} null hypothesis, such as $Y_i(1)=Y_i(0)$ for all $i$. In contrast, the design-based simulation described above uses permutations of the assignment to construct a DGP based on the observed outcomes in the application at hand, with the goal of evaluating the properties of an alternative inference procedure (e.g., a $t$-test with robust standard errors for hypotheses regarding the SATE).
\end{remark}

\subsection{A superpopulation framework}
\label{sec:rct-superpopulation}

Consider now a sampling framework in which the pairs $(Y_i(0),Y_i(1))$ are iid draws from a population and treatment is assigned independently of these draws, using the same allocation rule as in Section~\ref{sec:rct-finite}. In this case, both the sample of potential outcomes and treatment assignment are stochastic, and the target is the population ATE, $\tau=\mathbb E[Y_i(1)-Y_i(0)]$. The difference in means is unbiased for $\tau$, and inference using robust standard errors is asymptotically valid as $\min\{N_1,N-N_1\}\to\infty$, under standard regularity conditions.

As in Section~\ref{sec:rct-finite}, we can construct design-based simulations by holding the observed outcomes fixed and drawing new treatment assignments. Although the simulation DGP treats potential outcomes as fixed, we show that these simulations can still be informative about inference in the superpopulation framework, where potential outcomes are stochastic. Suppose, in particular, that the marginal distributions of $Y_i(1)$ and $Y_i(0)$ are the same, so the population ATE is zero, and let $R$ denote the event that the same nominal $\alpha$ test with robust standard errors rejects $H_0:\tau=0$ in the original experiment.

\begin{proposition}\label{prop:rct-assessment}

In the iid sampling framework described above, if $Y_i(1)\overset d=Y_i(0)$, then
$ \mathbb E[\Gamma_{\mathrm{rob}}^Y(\mathbf Y)]=\Pr(R)$.

\end{proposition}

The proof is in Appendix~\ref{app:rct-proof}. Under the stated condition, the $\mathbf{Y}$-fixed design-based simulation recovers the rejection probability of the original test on average over the observed outcome vectors. Thus these simulations can be informative about an inference procedure in a framework in which potential outcomes are stochastic, even though the DGP used in those simulations holds potential outcomes fixed.

The limitations discussed in Section~\ref{sec:rct-finite} nevertheless remain relevant. When the marginal distributions of $Y_i(1)$ and $Y_i(0)$ differ, these simulations need not recover the rejection probability of the original test, even if the population ATE is zero. In the next section, we develop these insights in a setting where design-based simulations have been widely used: assessing inference problems due to error dependence.

\section{Design-based simulations and error dependence}
\label{Sec: SS}\label{Sec: DB in SS}

We consider now a simple model in which treatment is assigned at the group level and use it to evaluate whether design-based simulations can reveal problems with robust standard errors due to dependence within groups. More specifically, we formally characterize what these simulations recover and compare alternative ways of constructing them. The results also provide insights into the use of design-based simulations in shift-share and DiD settings.

\subsection{A simple model with treatment assigned to groups}
\label{sec:group-model}

Consider a setting in which $N=Fm$ individuals are partitioned into $F$ groups of equal size $m \geq 2$, where $F$ is even. Exactly half of the groups receive treatment, and all allocations satisfying this constraint are equally likely. As in the superpopulation framework of Section~\ref{sec:rct-superpopulation}, potential outcomes and treatment allocation are stochastic, with treatment assigned independently of potential outcomes. Group membership is fixed. 
For individual $i$ in group $f$, potential outcomes are given by
\begin{equation}
 Y_{fi}(0)=\beta_0+u_{fi},
 \qquad
 Y_{fi}(1)=Y_{fi}(0)+\beta+\eta_f+\zeta_{fi}.
 \label{eq:potential-groups}
\end{equation}
We assume that the group vectors $(u_f,\eta_f,\zeta_f)$ are iid across groups, where $u_f=(u_{f1},\ldots,u_{fm})$ and $\zeta_f=(\zeta_{f1},\ldots,\zeta_{fm})$. The baseline errors $u_{fi}$ have mean zero, variance $\sigma_u^2$, and covariance $\rho$ between distinct individuals in the same group. When $\sigma_u^2>0$, the within-group correlation is $\rho/\sigma_u^2$.

The parameter $\beta$ is the population ATE, and we allow for two types of treatment-effect heterogeneity. Group-level treatment-effect heterogeneity is captured by $\eta_f$, which is independent of $u_f$, with mean zero and variance $\sigma_\eta^2$. Individual-level treatment-effect heterogeneity is captured by $\zeta_{fi}$, which are iid across individuals, independent of $(u_f,\eta_f)$, with mean zero and variance $\sigma_\zeta^2$.

Using the linking equation $Y_{fi}=D_fY_{fi}(1)+(1-D_f)Y_{fi}(0)$, the observed outcome is
\begin{equation}
 Y_{fi}=\beta_0+\beta D_f+\epsilon_{fi},
 \qquad \epsilon_{fi}=u_{fi}+D_f(\eta_f+\zeta_{fi}).
 \label{eq:model-groups}
\end{equation}

Let $\widehat\beta$ be the OLS coefficient from a regression of $Y_{fi}$ on an intercept and $D_f$, and let $\widehat\epsilon_{fi}=Y_{fi}-\widehat\beta_0-\widehat\beta D_f$ be its residuals. Write $\overline{\widehat\epsilon}_f=m^{-1}\sum_i\widehat\epsilon_{fi}$. The robust and cluster-robust variance estimators are, respectively,\footnote{We use versions without degrees-of-freedom corrections in this model. Corrections that converge to one do not change the asymptotic results, but can affect finite-sample results.}
\begin{equation}
 \widehat V_{\mathrm{rob}}=\frac4{N^2}\sum_{f,i}\widehat\epsilon_{fi}^{\,2},
 \qquad
 \widehat V_{\mathrm{cluster}}=\frac4{F^2}\sum_f\overline{\widehat\epsilon}_f^{\,2}.
 \label{eq:variance-estimators}
\end{equation}

Robust standard errors may fail in the group-assignment model because they do not account for dependence within groups. The following result characterizes the variance of $\widehat\beta$ and compares the asymptotic behavior of robust and cluster-robust variance estimators for inference about the population ATE.

\begin{proposition}\label{prop:original-ate}

Consider the group-assignment model described above. Assume finite fourth moments and $\sigma_u^2+(m-1)\rho+m\sigma_\eta^2/2+\sigma_\zeta^2/2>0$. The variance of $\widehat\beta$, over draws of potential outcomes and treatment assignment, is
\[
 \operatorname{Var}(\widehat\beta)
 =\frac4{Fm}\left[\sigma_u^2+(m-1)\rho+\frac{m\sigma_\eta^2}{2}+\frac{\sigma_\zeta^2}{2}\right].
\]

For an asymptotic sequence in which $F\to\infty$ and $m$ is fixed,
\[
 \sqrt N(\widehat\beta-\beta)
 \xrightarrow{d}
 N\!\left(0,\;4\left[\sigma_u^2+(m-1)\rho+\frac{m\sigma_\eta^2}{2}+\frac{\sigma_\zeta^2}{2}\right]\right),
\]

and the variance estimators satisfy
\begin{align*}
 N\widehat V_{\mathrm{rob}}
 &\xrightarrow{p}4\left[\sigma_u^2+\frac{\sigma_\eta^2+\sigma_\zeta^2}{2}\right],\\
 N\widehat V_{\mathrm{cluster}}
 &\xrightarrow{p}4\left[\sigma_u^2+(m-1)\rho+\frac{m\sigma_\eta^2}{2}+\frac{\sigma_\zeta^2}{2}\right].
\end{align*}

\end{proposition}

The proof is in Appendix~\ref{Appendix_design_based}. The robust variance estimator is asymptotically valid when $\rho=\sigma_\eta^2=0$, including when $\sigma_\zeta^2>0$. For nonnegative $\rho$, positive within-group covariance or group-level treatment-effect heterogeneity causes it to asymptotically understate the sampling variance of $\widehat\beta$, so a test using a standard normal critical value over-rejects. In contrast, the cluster-robust variance estimator is asymptotically valid under the assumptions of the proposition, including in the presence of within-group dependence and either form of treatment-effect heterogeneity.

This model directly describes a cluster-randomized experiment. It also corresponds to a stylized shift-share design in which each individual is exposed to one binary group-specific shock. In the usual shift-share notation, let $w_{ig}=1$ if individual $i$ belongs to group $g$, and $w_{ig}=0$ otherwise, and let the group shock be $\mathcal X_g=D_g$. Then the shift-share regressor $x_i=\sum_gw_{ig}\mathcal X_g$ equals $D_f$ for an individual in group $f$. This is a case with nonoverlapping exposures. In this nonoverlapping shift-share interpretation, group-cluster standard errors also correspond, up to finite-sample corrections, to the unrestricted variance formulas in \citet{Adao} and to the shock-level regression approach of \citet{Borusyak}.

Overall, there are two main reasons why an applied researcher may be concerned about inference. Robust inference can be distorted by within-group dependence even when the total number of observations is large. Cluster-robust inference allows for this dependence, but its asymptotic approximation may be poor when there are few independent groups.

\begin{remark}[Connection to DiD]\label{remark:did}
Although this framework does not directly describe a DiD model with serial dependence, Appendix~\ref{app:did} considers a stylized DiD model with AR(1) errors and a common treatment date. We obtain analogous results for inference and for the two simulation constructions in that model.
\end{remark}

\subsection{\texorpdfstring{$\mathbf Y$-fixed simulations to assess robust inference}{Y-fixed simulations to assess robust inference}}
\label{sec:outcome-fixed}

Consider now design-based simulations with new draws of $\mathbf D^*=(D_1^*,\ldots,D_F^*)$ from the balanced assignment distribution, while holding the vector of observed outcomes, $\mathbf Y$, fixed. For each draw, we run the regression of $Y_{fi}$ on an intercept and $D_f^*$ and test the null of zero effect using robust standard errors and a standard normal critical value. These simulations set $\widetilde Y_{fi}(0)=\widetilde Y_{fi}(1)=Y_{fi}$. Therefore, the null is true in the simulation DGP, and the OLS estimator in the simulated regression is unbiased for zero over treatment assignments.

Conditional on the observed outcome vector $\mathbf Y$, uncertainty in the $\mathbf{Y}$-fixed design-based simulations comes only from realizations of $\mathbf D^*$. A superscript $*$ on probabilities and expectations refers to this assignment randomness, conditional on the original data. Let $\widehat\gamma^*(\mathbf Y)$ denote the OLS estimator in a simulated regression, and let $\widehat V_{\mathrm{rob}}^*(\mathbf Y)$ denote the robust variance estimator from that regression. Using the notation introduced in Section~\ref{sec:rct-finite}, the rejection probability in the simulation is
\[
 \Gamma_{\mathrm{rob}}^Y(\mathbf Y)=\mathbb P^*\!\left(
 \frac{|\widehat\gamma^*(\mathbf Y)|}{\sqrt{\widehat V_{\mathrm{rob}}^*(\mathbf Y)}}
 >z_{1-\alpha/2}\,\middle|\,\mathbf Y\right),
\]
where $\Phi$ is the standard normal distribution function and $z_{1-\alpha/2}=\Phi^{-1}(1-\alpha/2)$, with $0<\alpha<1$. As in Section~\ref{sec:rct-finite}, we focus on this exact probability rather than its Monte Carlo estimate, $ \widehat{\Gamma}_{\mathrm{rob},B}^Y(\mathbf Y)$.

The following proposition characterizes the asymptotic behavior of $\Gamma_{\mathrm{rob}}^Y(\mathbf Y)$ as a function of the parameters of the original DGP.

\begin{proposition}\label{Prop}\label{prop:outcome-fixed}

Under the assumptions and asymptotic sequence of Proposition~\ref{prop:original-ate}, the rejection rate in the $\mathbf Y$-fixed design-based simulation satisfies
\begin{equation}
 \Gamma_{\mathrm{rob}}^Y(\mathbf Y)
 \xrightarrow{p}
 2\left[1-\Phi\!\left(z_{1-\alpha/2}\sqrt{
 \frac{\sigma_u^2+(\sigma_\eta^2+\sigma_\zeta^2)/2+\beta^2/4}
 {\sigma_u^2+(m-1)\rho+m\sigma_\eta^2/2+\sigma_\zeta^2/2+m\beta^2/4}}
 \right)\right].
 \label{eq:gamma-outcome}
\end{equation}

For $\rho\geq0$, this limit equals $\alpha$ if and only if $\beta=\rho=\sigma_\eta^2=0$, and exceeds $\alpha$ otherwise, including when $\rho=\sigma_\eta^2=0$ but $\beta\neq0$.

\end{proposition}

The proof is in Appendix~\ref{Appendix_design_based}.

First, consider a setting with no group-level treatment-effect heterogeneity and no within-group covariance, so $\sigma_\eta^2=\rho=0$. Robust inference is then asymptotically valid in the original experiment. However, when $\beta\neq0$, the limiting rejection rate in \eqref{eq:gamma-outcome} exceeds $\alpha$. Therefore, the $\mathbf Y$-fixed simulations indicate over-rejection, incorrectly suggesting that robust standard errors are invalid due to error dependence. The fundamental difficulty is that $Y_{fi}$ contains the component $\beta D_f$, which is common to all individuals in a group. This component does not belong to the error in the original regression. In the simulation, however, the original treatment is no longer the regressor, while its effect remains in the fixed outcome. The fixed outcome in the $\mathbf Y$-fixed design-based simulations therefore retains a component common to individuals in the same group, even when the original untreated outcomes are independent.

Overall, Proposition~\ref{prop:outcome-fixed} shows that $\mathbf Y$-fixed design-based simulations can overstate the over-rejection problem when the treatment effect is nonzero. In particular, they can indicate over-rejection even in settings where robust standard errors are asymptotically valid.

\subsection{Alternatives for assessing robust inference}
\label{sec:alternatives}

We consider two alternative ways of constructing these simulations. The first is to use an outcome from a period before treatment, for which we expect the treatment effect to be absent and the distribution to be the same as, or similar to, that of untreated outcomes in the current application. To state the result precisely, suppose the pre-treatment group outcome vectors have the same joint distribution as the current untreated group outcome vectors, up to a common level that does not affect the test, and satisfy the same independence and moment conditions, with $\sigma_u^2+(m-1)\rho>0$. Proposition~\ref{prop:outcome-fixed} then implies
\[
 \Gamma_{\mathrm{rob}}^Y(\mathbf Y^{\mathrm{pre}})
 \xrightarrow{p}
 2\left[1-\Phi\!\left(z_{1-\alpha/2}
 \sqrt{\frac{\sigma_u^2}{\sigma_u^2+(m-1)\rho}}\right)\right].
\]

For $\rho\geq0$, these simulations have a rejection probability that converges to $\alpha$ when $\rho=0$ and to a value above $\alpha$ when $\rho>0$. This is precisely the behavior we would want from an assessment designed to reveal inference problems due to positive within-group dependence in the untreated potential outcomes. They can therefore reveal such dependence without being contaminated by the effect of the original treatment.

With a similar rather than identical distribution, these simulations may still provide a useful approximation, depending on whether the differences affect features relevant for inference. An important limitation, however, is that this strategy does not capture the additional uncertainty generated by group-level treatment-effect heterogeneity, since we only observe untreated potential outcomes. Therefore, if robust inference for the ATE over-rejects because of group-level treatment-effect heterogeneity, a $\mathbf Y$-fixed design-based simulation using pre-treatment outcomes would not reveal this problem.

The second alternative is to remove the estimated treatment effect before fixing outcomes. Since each simulated regression includes an intercept, fixing $Y_{fi}-\widehat\beta D_f$ gives the same coefficients and standard errors as using the unrestricted OLS residuals $\widehat\epsilon_{fi}$ as the simulated potential outcomes. We refer to these exercises as \emph{$\boldsymbol\epsilon$-fixed design-based simulations}. \citet{Borusyak} use this adjustment in their Appendix A.11. Let $\widehat{\boldsymbol\epsilon}=(\widehat\epsilon_{fi})_{f,i}$ denote the vector of unrestricted OLS residuals. The $\boldsymbol\epsilon$-fixed assessment is

\[
 \Gamma_{\mathrm{rob}}^\epsilon(\mathbf Y,\mathbf D)
 =\mathbb P^*\!\left(
 \frac{|\widehat\gamma^*(\widehat{\boldsymbol\epsilon})|}{\sqrt{\widehat V_{\mathrm{rob}}^*(\widehat{\boldsymbol\epsilon})}}
 >z_{1-\alpha/2}\,\middle|\,\mathbf Y,\mathbf D\right).
\]

Unlike the $\mathbf Y$-fixed assessment, this function also depends on the original assignment through the estimated residuals.

\begin{proposition}\label{prop:residual-fixed}

Under the assumptions and asymptotic sequence of Proposition~\ref{prop:original-ate}, the rejection rate in the $\boldsymbol\epsilon$-fixed design-based simulation satisfies
\begin{equation}
 \Gamma_{\mathrm{rob}}^\epsilon(\mathbf Y,\mathbf D)
 \xrightarrow{p}
 2\left[1-\Phi\!\left(z_{1-\alpha/2}\sqrt{
 \frac{\sigma_u^2+(\sigma_\eta^2+\sigma_\zeta^2)/2}
 {\sigma_u^2+(m-1)\rho+m\sigma_\eta^2/2+\sigma_\zeta^2/2}}
 \right)\right].
 \label{eq:gamma-residual}
\end{equation}

This is also the limiting rejection probability of the robust test of the true population ATE in the original experiment. For $\rho\geq0$, it equals $\alpha$ if and only if $\rho=\sigma_\eta^2=0$, and exceeds $\alpha$ otherwise, regardless of $\beta$.

\end{proposition}

The proof is in Appendix~\ref{Appendix_design_based}. Subtracting the estimated effect eliminates the contribution of the common treatment effect to the simulation second moments as $\widehat\beta\xrightarrow{p}\beta$. The $\boldsymbol\epsilon$-fixed simulations therefore recover the distortion relevant for the original regression even when the true treatment effect is nonzero. In particular, when $\rho=\sigma_\eta^2=0$, they do not indicate asymptotic over-rejection merely because $\beta\neq0$.

When treatment effects vary across groups, subtracting the estimated average effect leaves $\eta_fD_f$ in the adjusted outcome. This component creates dependence within groups, but it also contributes to uncertainty about the population ATE in the original experiment. Therefore, it is appropriate that the $\boldsymbol\epsilon$-fixed assessment retains it. The individual component $\zeta_{fi}D_f$ also remains, but individual-level treatment-effect heterogeneity does not by itself lead to asymptotic over-rejection. The two alternatives can consequently provide different information: pre-treatment outcomes can help assess dependence in untreated outcomes, while adjusted post-treatment outcomes also reflect uncertainty from heterogeneous treatment effects.

Finally, we note that at first glance it may seem surprising that design-based simulations can be informative about characteristics of the error distribution, since they hold potential outcomes, and therefore the realized errors, fixed. Nevertheless, the results above show that these simulations can recover the implications of error dependence for inference as the number of independent groups increases. This happens because the realized outcome vector can retain information about the distribution that generated it, even though the simulations do not draw new errors.

\subsection{Finite-sample interpretation}
\label{sec:finite-assessment}

Design-based simulations can also be informative about finite-sample distortions. To extend Proposition~\ref{prop:rct-assessment} to group assignment, retain the balanced group-assignment design but allow a general joint distribution with finite means for the potential-outcome vectors $\mathbf Y_f(d)=(Y_{f1}(d),\ldots,Y_{fm}(d))'$, rather than imposing the additive structure in \eqref{eq:potential-groups}. For a fixed inference method $j$, let $\Gamma_j^Y(\mathbf Y)$ be its rejection probability in the $\mathbf Y$-fixed simulation and $R_j$ the event that the same test rejects zero in the original experiment.

\begin{proposition}\label{prop:finite-group-assessment}
Suppose that the pairs $(\mathbf Y_f(0),\mathbf Y_f(1))$ are independent across groups and independent of treatment assignment. If $\mathbf Y_f(1)\overset d=\mathbf Y_f(0)$ for every group, then, using the true assignment distribution in the simulations, $\mathbb E[\Gamma_j^Y(\mathbf Y)]=\Pr(R_j)$.
This equality holds for inference based on either robust or cluster-robust standard errors.
\end{proposition}
The proof is in Appendix~\ref{app:finite-group-assessment}. As in Proposition~\ref{prop:rct-assessment}, design-based simulations recover the true rejection probability on average under the strong condition that the treatment and control potential-outcome vectors have the same distribution.

More generally, averaging the $\mathbf Y$-fixed assessment over original samples corresponds to an artificial DGP with $\widetilde{\mathbf Y}(1)\overset d=\widetilde{\mathbf Y}(0)$, where the common distribution of these potential-outcome vectors depends on the marginal distributions of the treatment and control potential outcomes and on the original assignment mechanism. When these marginal distributions differ, this artificial DGP need not reproduce the finite-sample properties of the original experiment. The limitations discussed in Section~\ref{Sec: DB in RCT} therefore remain relevant.

An important distinction, however, is that the treatment-effect problem highlighted in Proposition~\ref{prop:outcome-fixed} does not create the same additional problem for cluster-robust inference in the group-assignment model. The component $\beta D_f$ retained in the $\mathbf Y$-fixed simulations varies only at the group level, and cluster-robust inference already allows arbitrary dependence within groups. Thus, in this setting, retaining a homogeneous treatment effect does not by itself create the same spurious indication of invalid inference. As such, if a $\mathbf Y$-fixed design-based simulation reveals large distortions in cluster-robust inference, this indicates poor finite-sample performance of the cluster-robust asymptotic approximation in the simulation DGP. Importantly, this DGP preserves several features of the original application that may matter for the quality of this approximation, including the number of observations, the number and sizes of the groups, and the covariate structure. However, the caveat discussed in the previous paragraph still applies: the quality of the asymptotic approximation may also depend on the true distribution of the potential outcomes, which need not be reproduced by these design-based simulations.

Finally, we note that the finite-sample property from Proposition~\ref{prop:finite-group-assessment} does not generally hold for $\boldsymbol\epsilon$-fixed simulations, even with homogeneous effects. Subtracting the estimated effect removes not only the true treatment component but also the realized difference in average baseline errors between the two treatment arms. The assessment is then invariant to the value of the homogeneous effect, but its expectation need not equal the original rejection probability. Since this is a finite-sample conclusion, this does not contradict the recovery result as $F\to\infty$ in Proposition~\ref{prop:residual-fixed}.

\section{Using assessments to choose an inference method}
\label{sec:selection}

The model in Section~\ref{sec:group-model} illustrates a trade-off faced by applied researchers. Robust standard errors can be distorted by dependence within groups, while cluster-robust methods may perform poorly in finite samples when there are few independent groups. We consider whether design-based simulations provide sufficiently reliable information to choose between these methods.

Consider a decision rule under which we use robust inference when its design-based assessment is below a cutoff $\kappa>0$. For example, we may use robust standard errors whenever the assessment is below $0.10$. An important concern with using an assessment to validate an inference method is the possibility of pre-testing distortions \citep{Ferman_assessment,Andrews2018,Guggenberger2010,Roth2022}. In the group-assignment model of Section~\ref{sec:group-model}, a $\mathbf Y$-fixed assessment depends on the component of the outcome associated with the estimated treatment effect. As we illustrate in Section~\ref{Sec: simulations}, selecting samples on the basis of a low assessment can favor estimates closer to zero and induce bias when $\beta\neq0$. Thus, an estimator that is unbiased unconditionally may become biased conditional on a favorable assessment.

In contrast, an $\boldsymbol\epsilon$-fixed design-based assessment offers some protection against pre-testing distortions under additional assumptions. In Appendix~\ref{app:selection-proofs}, we show that the estimator remains unbiased conditional on $\Gamma_{\mathrm{rob}}^\epsilon(\mathbf Y,\mathbf D)<\kappa$ under a symmetry condition. Under a Gaussian version of the model, conditioning on this event also preserves the estimator's entire distribution. We illustrate this difference in Section~\ref{Sec: simulations}. Preserving these properties of the coefficient estimator does not, however, establish the validity of the selected inference procedure.

Another concern is whether the assessment has a high probability of detecting relevant distortions. Proposition~\ref{prop:residual-fixed} provides a positive result: as the number of groups increases, the residual-based assessment converges to the limiting rejection probability of the original test. With a fixed cutoff above the nominal level, it therefore detects distortions whose limiting rejection probability exceeds that cutoff with probability approaching one and approves robust inference with probability approaching one when it is asymptotically valid. In this setting, however, cluster-robust inference is also asymptotically valid. With many independent groups, it is consequently less important to determine whether robust inference assuming independence within groups would be reliable.

With few groups, determining whether we can rely on robust inference becomes more useful because cluster-robust inference may itself have poor finite-sample properties. Unfortunately, the same small number of independent groups can make the assessment less informative about dependence. In particular, a favorable assessment may occur with high probability even when robust standard errors fail to account for relevant dependence in the original DGP. In Sections~\ref{Sec: simulations} and~\ref{Sec: SS illustrations}, we show that this can be an important concern in practice.

One possibility for addressing this limitation is to consider additional pre-treatment outcomes. Suppose that we observe a vector of $mF$ outcomes, $\mathbf Y_t$, for periods $t=1,\ldots,T$, where $\mathbf Y_T$ is the realized outcome of the experiment and $\mathbf Y_t$ for $t<T$ contains pre-treatment outcomes. 
Potential outcomes in period $T$ follow the model in \eqref{eq:potential-groups}:
\[
Y_{fiT}(0)=\beta_0+u_{fiT},
\qquad
Y_{fiT}(1)=Y_{fiT}(0)+\beta+\eta_{fT}+\zeta_{fiT}.
\]
For $t<T$, observed outcomes are $Y_{fit}=\beta_0+u_{fit}$. Assume first that the vectors $\mathbf u_t=(u_{fit})_{f,i}$ are iid across periods, with the within-period distribution specified in Section~\ref{sec:group-model}, and independent of treatment assignment and the treatment-effect components.
 We construct a $\mathbf Y$-fixed assessment separately for each pre-treatment outcome vector, $\Gamma_{\mathrm{rob}}^Y(\mathbf Y_t)$, retaining the current assignment distribution and inference procedure. With $F$ fixed, a law of large numbers gives
\[
\frac{1}{T-1}\sum_{t=1}^{T-1}
\Gamma_{\mathrm{rob}}^Y(\mathbf Y_t)
\xrightarrow{p}
\mathbb E\!\left[\Gamma_{\mathrm{rob}}^Y(\mathbf Y_1)\right]
\qquad\text{as }T\to\infty.
\]
When treatment effects are homogeneous, this expectation equals the finite-sample rejection probability of robust inference for the true effect in the current application. Thus additional pre-treatment outcomes can make the assessment more informative even when the number of groups remains small. If this rejection probability exceeds the approval cutoff, the probability that the average assessment approves robust inference converges to zero. When effects are heterogeneous, the historical outcomes still provide information about spatial dependence in untreated outcomes, but do not reproduce the additional uncertainty generated by treatment-effect heterogeneity, including the within-group dependence induced by $\eta_{fT}$.

The iid assumption is stronger than necessary for this positive result. The same convergence result holds if the outcome process has a stable marginal distribution and satisfies conditions for a law of large numbers for the assessments, for example if it is stationary and ergodic. Serial dependence, however, raises an additional issue when the available history is finite. In this case, the current estimator need not be independent of the pre-treatment assessments, so conditioning on favorable historical assessments may change its distribution. Using $\boldsymbol\epsilon$-fixed simulations for the historical outcomes may provide safeguards against such selection effects under additional conditions on the joint distribution across periods. At the same time, removing an estimated placebo effect from each historical outcome can change the finite-sample rejection probability recovered by the average assessment when $F$ is fixed.

Finally, the historical outcomes must be comparable to the current untreated outcomes in the dimensions relevant for inference. If their spatial dependence differs from that in the current application, averaging more historical assessments need not provide an accurate assessment of the inference problem of interest.

\Needspace{5\baselineskip}
\section{Monte Carlo simulations}
\label{Sec: simulations}

We present a Monte Carlo exercise illustrating the mechanisms from Section~\ref{Sec: DB in SS} and the limitations discussed in Section~\ref{sec:selection}. Table~\ref{tab:dbs-unified} considers simulations of the model presented in Section~\ref{sec:group-model}, with $F=4,8,16,100$ and $m=50$. Appendix~\ref{Appendix_permutation_shocks} describes the DGPs and implementation in detail. Columns (1) and (2) report rejection probabilities for tests of the true population ATE using robust and cluster-robust inference at a $5\%$ nominal level. Columns (3) and (4) report approval probabilities under the rule $\widehat\Gamma_B<0.10$, where $\widehat\Gamma_B$ is the fraction of simulated regressions that reject zero using robust standard errors, with the construction indicated by the column. Columns (5) and (6) report the bias of $\widehat\beta$ conditional on approval.

The results confirm the main implications of Propositions~\ref{prop:outcome-fixed} and~\ref{prop:residual-fixed}. Panels A and B have independent errors and differ only in the treatment effect. Consider first the results with 100 groups, for which robust inference has a true rejection probability of approximately $5\%$ in both panels. When $\beta=0.3$, the $\mathbf Y$-fixed assessment approves robust inference in only about $0.2\%$ of samples, compared with $99.50\%$ when $\beta=0$. Thus retaining the treatment effect can make the simulations indicate a problem even when robust inference is approximately valid, as predicted by Proposition~\ref{prop:outcome-fixed}. In contrast, the residual-based assessment approves in $99.56\%$ of samples in both panels, consistent with Proposition~\ref{prop:residual-fixed}. For the smaller group counts considered, false alarms are more frequent than with 100 groups, although the residual-based assessment still approves in approximately $89.6\%$ of samples with four groups in both panels.

Panels C and D introduce within-group dependence, with $\omega=0.15$ held fixed across group counts. Robust inference then rejects the true hypothesis in approximately $17$--$18\%$ of samples. Panel E instead introduces group-level treatment-effect heterogeneity while keeping untreated outcomes independent, also leading to approximately $17$--$18\%$ rejection rates when relying on robust inference. With 100 groups, both assessments rarely approve robust inference in these three scenarios, consistent with Propositions~\ref{prop:outcome-fixed} and~\ref{prop:residual-fixed}, since their limiting rejection probabilities exceed the approval cutoff.

The results with fewer groups, however, illustrate the limitation discussed in Section~\ref{sec:selection}. In Panel D with four groups, cluster-robust inference has a true rejection probability of approximately $29.6\%$. This is therefore a setting in which it would be particularly useful to determine whether robust inference is reliable. However, the residual-based assessment approves robust inference in approximately $66.7\%$ of samples, even though the robust test rejects the true null in approximately $17.6\%$ of samples at a nominal $5\%$ level. The probability of failing to detect the problem remains substantial with eight and sixteen groups, at approximately $40.7\%$ and $24.2\%$, respectively.

The final two columns illustrate the difference between the selection properties of the two assessments. Panels A and C confirm the intuition that, with a finite number of groups and $\beta\neq0$, conditioning on a favorable $\mathbf Y$-fixed assessment can bias $\widehat\beta$ toward zero. In contrast, the conditional biases for the $\boldsymbol\epsilon$-fixed assessment are close to zero, consistent with the result in Appendix~\ref{app:selection-proofs} and the symmetry conditions satisfied by the Gaussian DGPs used in these simulations.

\section{Empirical illustrations}
\label{Sec: SS illustrations}

We illustrate the use of design-based simulations in three shift-share applications, based on \citet{Autor}, \citet{Acemoglu}, and \citet{Dix}. These applications illustrate both the information that the assessments can provide and the difficulties involved in using them to compare inference methods.

\subsection{Shift-share framework}\label{Sec: SS setting}

Shift-share designs study the impact of a set of shocks on units differentially exposed to them. Consider $i=1,\ldots,N$ regions subject to $f=1,\ldots,F$ aggregate shocks $\mathcal X_f$. The shift-share variable is $x_i=\sum_{f=1}^Fw_{if}\mathcal X_f$, where the shares $w_{if}$ reflect how each shock affects each region. A simplified regression model is
\[
 y_i=\beta_0+\beta x_i+\epsilon_i.
\]

 A key insight from \citet{Adao} is that regions with similar shares may have correlated errors, leading to over-rejection if the inference method does not allow for this dependence. In particular, clustering by a geographic partition may not account for correlation across regions in different clusters but with similar exposure to the shocks.

\citet{Adao} and \citet{Borusyak} develop inference methods in a framework in which shocks are stochastic while potential outcomes and shares are conditioned on.\footnote{See \citet{GP} for an alternative framework for shift-share designs. In the homogeneous linear-effect model considered here, the parameter $\beta$ does not change when we condition on potential outcomes. Validity conditional on those outcomes can therefore also justify inference when both outcomes and shocks are stochastic, under the corresponding asymptotic conditions.} Their asymptotic results require sufficiently many independent shocks, or independent clusters of shocks, with no small number of them dominating the variation.

The model in Section~\ref{Sec: DB in SS} is a particular case with nonoverlapping shares and balanced binary shocks, providing a simple setting in which to study these design-based simulations. The main intuitions developed in that section have a natural analogue in more general shift-share designs. In this setting, the inference procedures of \citet{Adao} play a role similar to cluster-robust inference in the stylized model, in that they allow for the dependence induced by common exposure to shocks. By contrast, heteroskedasticity-robust standard errors, or standard errors clustered at the regional level, play a role similar to robust inference in the stylized model as they fail to account for dependence across regions with similar shock exposures.

In empirical shift-share applications, another important difference is that we generally do not know the exact distribution of the shocks. Following \citet{Adao}, we construct $\mathbf{Y}$-fixed design-based simulations by holding outcomes and shares fixed and drawing iid standard normal shocks. For each draw, we construct the shift-share regressor, estimate the specification, and test the null of zero effect. We also consider $\boldsymbol{\epsilon}$-fixed design-based simulations using the same distribution of shocks. The assumed shock distribution therefore provides an additional dimension along which the simulation DGP may differ from the original DGP. We use these simulations first to assess cluster-robust inference and then to assess the shift-share inference procedures.

\subsection{Assessing cluster-robust inference}
\label{sec:assessing-cluster-inference}

We first consider the use of geographically clustered standard errors in these applications. In the two U.S. applications \citep{Autor,Acemoglu}, standard errors are clustered at the state level, while in the Brazilian application \citep{Dix} they are clustered at the mesoregion level. These standard errors may be problematic if there is spatial correlation beyond the chosen clusters or if the number of clusters is too small for the asymptotic approximation to be accurate. \citet{Ferman_assessment} finds that, in the absence of spatial correlation, these approximations are less reliable for the weighted OLS specifications in these applications. For Table~\ref{Table_SS}, we therefore use unweighted specifications and cluster-robust standard errors with an HC3 correction.\footnote{\label{Footnote_dist}To distinguish spatial dependence from poor finite-cluster approximations, one can also consider simulations in which both shocks and errors are stochastic, drawing errors independently while retaining other relevant features of the specification. Such an exercise assesses performance without the spatial dependence present in the fixed outcome vector.}

We begin with the $\mathbf Y$-fixed design-based simulations used by \citet{Adao} to illustrate the importance of spatial correlation in these applications. Panel A of Table~\ref{Table_SS} reports rejection frequencies between $31.3\%$ and $41.2\%$ for the main outcomes. These results suggest substantial distortions in inference based on geographically clustered standard errors. However, Proposition~\ref{prop:outcome-fixed} shows that $\mathbf Y$-fixed simulations can overstate such distortions because the fixed outcome retains the effect of the original shocks. The rejection frequencies in Panel A should therefore be interpreted with this caveat in mind.

We next consider the two alternative constructions discussed in Section~\ref{sec:alternatives}. Panel B reports $\boldsymbol\epsilon$-fixed simulations, which remove the estimated effect before fixing the outcome. The placebo columns use pre-treatment outcomes, which are not affected by the shocks whose effects are estimated in the main specifications. Panel A fixes these outcomes directly, while Panel B first removes the estimated placebo effect.

For the China-shock application, both alternatives continue to indicate substantial distortions. The $\boldsymbol\epsilon$-fixed rejection frequency is $31.7\%$ for the main outcome. For the pre-treatment outcome, the rejection frequencies are $23.7\%$ for the $\mathbf Y$-fixed assessment and $22.6\%$ for the $\boldsymbol\epsilon$-fixed assessment. These values are below the $41.2\%$ rejection frequency from the original $\mathbf Y$-fixed assessment, consistent with the treatment effect accounting for part of the dependence in the fixed outcome. Nevertheless, the alternative assessments continue to provide evidence that spatial dependence not captured by state-level clustering can generate relevant distortions in this application.

The evidence is less clear for the exposure-to-robots application. The $\boldsymbol\epsilon$-fixed assessment for the main outcome remains high, at $25.6\%$, whereas the assessments based on the pre-treatment outcome are $5.9\%$ for the $\mathbf Y$-fixed construction and $4.9\%$ for the $\boldsymbol\epsilon$-fixed construction. One possible interpretation is that there is little spatial dependence in the untreated potential outcomes, but treatment-effect heterogeneity generates spatial dependence after treatment. This would be consistent with the difference between the two assessments, since the $\boldsymbol\epsilon$-fixed construction retains heterogeneous treatment effects whereas the pre-treatment outcome does not. Taken at face value, this pattern is consistent with state-level clustering being adequate for testing a sharp null of no treatment effects, but not for inference on the shift-share estimand when treatment effects are heterogeneous. For the trade-liberalization application, the alternative assessments are even smaller, providing some evidence that spatial correlation may not be a first-order concern in this application: the $\boldsymbol\epsilon$-fixed rejection frequency is $1.2\%$ for the main outcome, and both constructions yield rejection frequencies close to zero for the pre-treatment outcome.

As emphasized in Sections~\ref{sec:selection} and~\ref{Sec: simulations}, however, these assessments may have a low probability of detecting relevant spatial correlation in some applications. Appendix~\ref{Appendix_SS_applications} considers a simple exercise to evaluate this possibility while preserving the observed share structure of each empirical application. The exercise extends the logic of Section~\ref{Sec: simulations} to the more complex dependence structure of shift-share designs. For each application, we use the observed shares to generate outcome components with different degrees of spatial dependence. For each resulting DGP, we calculate both the true rejection probability of geographically clustered inference and the probability that the design-based assessment is at least $10\%$. Both the original tests and the design-based assessments use the same cluster-HC3 variance estimator and standard normal critical values as Table~\ref{Table_SS}. We can therefore evaluate how the probability of detecting a problem varies with the magnitude of the actual size distortion, given the share structure of each application.

The results in Figure~\ref{Fig_SS_flagging} differ substantially across applications. For the China-shock application, when there is no spatial correlation, the probability that the assessment is at least $10\%$ is $5.8\%$ for the $\mathbf Y$-fixed assessment and $4.7\%$ for the $\boldsymbol\epsilon$-fixed assessment. When spatial correlation raises the true rejection probability to around $15\%$, detection probabilities are around $80\%$, and exceed $90\%$ when the true rejection probability is around $17\%$. Thus, in these auxiliary DGPs, the assessments rarely flag a problem in the absence of spatial correlation, but their probability of detecting a problem increases rapidly as spatial correlation generates larger size distortions.

For the exposure-to-robots application, a true rejection probability of around $15\%$ is associated with detection probabilities of around $60\%$ for the $\mathbf Y$-fixed assessment and $46\%$ for the $\boldsymbol\epsilon$-fixed assessment. The limitation is even more pronounced for the trade-liberalization application: when the true rejection probability is around $15\%$, detection probabilities are only around $35\%$ and $13\%$, respectively. These auxiliary simulations illustrate why a pre-treatment assessment close to or below $5\%$ need not rule out relevant spatial dependence in untreated potential outcomes.

Overall, the alternative assessments continue to indicate an important spatial-correlation problem in the China-shock application, although smaller than suggested by the original $\mathbf Y$-fixed exercise. For the other two applications, the evidence is more favorable to geographically clustered inference, but the auxiliary simulations show that low assessments can coexist with relevant distortions. These results illustrate an important warning from Sections~\ref{sec:selection} and~\ref{Sec: simulations}: a design-based assessment that does not detect a problem need not provide strong evidence that the inference procedure is reliable.

\subsection{Assessing shift-share-specific inference methods}
\label{Sec: assessing approximations}

Panel C of Table~\ref{Table_SS} uses $\mathbf Y$-fixed design-based simulations to assess the inference procedures from \citet{Adao}. The interpretation of these assessments differs from the interpretation of the $\mathbf Y$-fixed assessment for geographically cluster-robust inference. The procedures from \citet{Adao} are designed to allow for dependence across regions generated by common exposure to shocks. Therefore, retaining a treatment component in the fixed outcome does not create the same problem highlighted in Proposition~\ref{prop:outcome-fixed}. Large rejection frequencies instead should indicate poor finite-sample performance of the corresponding asymptotic approximation in the simulation DGP. As discussed in Section~\ref{sec:finite-assessment}, while the DGPs in these simulations need not reproduce the true distribution of potential outcomes or shocks, they preserve several features of the original application that can matter for that approximation, including the number of observations, the covariates and structure of shares, and the number of shocks.

For the China-shock application, the assessments provide little evidence of finite-sample distortions in these inference procedures. For the other two applications, however, the design-based simulations suggest that these procedures can lead to large distortions, with rejection frequencies reaching $54.7\%$. The contrast across applications is consistent with the role of the number of shocks in the asymptotic approximations. The China-shock application has 395 sectors, while the exposure-to-robots and trade-liberalization applications have 19 and 20 sectors, respectively. The assessments are close to the nominal level in the first application and much larger in the latter two.

\subsection{What do the assessments suggest about inference?}

Taken together, the results illustrate that design-based assessments can provide useful information about the relative strengths and limitations of alternative inference procedures, but also have intrinsic limitations that can make their interpretation difficult.

For the China-shock application, the different pieces of evidence point in the same direction. The alternative assessments for state-clustered inference continue to indicate relevant distortions from spatial dependence beyond states, while the assessments of the procedures from \citet{Adao} are close to the nominal level. Thus the simulations provide strong evidence in favor of using an inference procedure that accommodates the dependence generated by common exposure to shocks in this application, such as the procedures suggested by \citet{Adao} or \citet{Borusyak}.

The situation is more difficult in the exposure-to-robots and trade-liberalization applications. In both cases, the assessments of the procedures from \citet{Adao} show relevant over-rejection, raising concerns about their finite-sample approximations. It therefore becomes particularly important to know whether geographically clustered inference, which does not account for spatial dependence across geographical clusters induced by common exposure to shocks, provides a reliable alternative. Taken at face value, the design-based assessments are more favorable to geographically clustered inference in these applications.

However, this is where the limitation emphasized in Sections~\ref{sec:selection} and~\ref{Sec: simulations} becomes particularly relevant. When the assessments suggest that inference procedures designed to accommodate the relevant spatial dependence have poor finite-sample approximations, it becomes more valuable to know whether a simpler inference procedure based on geographical clustering is reliable. Yet, it is precisely in these settings that the design-based assessments used to answer this question can have a substantial probability of failing to detect relevant spatial-correlation problems. Thus, although the assessments provide some evidence in favor of geographically clustered inference in these applications, they do not provide strong assurance that the resulting tests have size close to the nominal level.

Another possibility in such settings is to use inference procedures that allow for dependence induced by common shocks while providing finite-sample validity under stronger assumptions about the assignment or shock process. Randomization-inference approaches provide examples of such procedures \citep{Alvarez,Borusyak2}. Their usefulness, however, depends on whether the assumptions on the shock-assignment mechanism required for their validity are credible in the application.

\section{Conclusion}\label{sec:conclusion}

This paper studies when design-based simulations are informative about inference procedures and when their interpretation can be misleading. Their informativeness depends on how the DGP used in the simulations relates to the original DGP along dimensions that matter for inference. In a model with treatment assigned to groups, holding observed outcomes fixed can confound a nonzero treatment effect with error dependence and overstate the distortions of inference procedures that do not account for such dependence.

At the same time, design-based simulations can be informative about characteristics of the error distribution even though they hold the realized errors fixed. The realized outcome vector contains information about the distribution that generated it. Removing an estimated treatment effect or using an appropriate pre-treatment outcome can therefore provide useful alternative assessments. These alternatives, however, provide different information: pre-treatment outcomes are informative about dependence in untreated potential outcomes, while residual-based simulations can also retain uncertainty generated by heterogeneous treatment effects.

We also show that using these simulations to evaluate and select inference procedures in specific empirical applications requires additional care. In particular, these assessments may be less reliable precisely in settings where they would be most useful. When there are few independent groups or shocks, inference procedures that accommodate richer forms of dependence may have poor finite-sample properties, but the design-based simulations used to assess whether simpler procedures are reliable may have a low probability of detecting relevant dependence in the errors. Our empirical illustrations with shift-share designs illustrate exactly this tension.

\singlespace

\clearpage
\begin{table}[p]
\centering
\renewcommand{\baselinestretch}{1}\small
\caption{\textbf{Design-based simulations and method approval}}
\label{tab:dbs-unified}
\setlength{\tabcolsep}{0pt}
\setlength{\doublerulesep}{1.4pt}
\renewcommand{\arraystretch}{1.14}
\begin{tabular}{@{}
 >{\centering\arraybackslash}p{0.055\linewidth}@{\hspace{0.025\linewidth}}
 >{\centering\arraybackslash}p{0.130\linewidth}@{\hspace{0.010\linewidth}}
 >{\centering\arraybackslash}p{0.130\linewidth}@{\hspace{0.035\linewidth}}
 >{\centering\arraybackslash}p{0.130\linewidth}@{\hspace{0.010\linewidth}}
 >{\centering\arraybackslash}p{0.130\linewidth}@{\hspace{0.035\linewidth}}
 >{\centering\arraybackslash}p{0.130\linewidth}@{\hspace{0.010\linewidth}}
 >{\centering\arraybackslash}p{0.130\linewidth}@{}}
\hline\hline
\rule{0pt}{2.7ex}
 & \multicolumn{2}{c@{\hspace{0.035\linewidth}}}{Test size}
 & \multicolumn{2}{c@{\hspace{0.035\linewidth}}}{Approval probability}
 & \multicolumn{2}{c}{Conditional bias} \\
\cmidrule(l{0pt}r{0.035\linewidth}){2-3}\cmidrule(l{0pt}r{0.035\linewidth}){4-5}\cmidrule(l{0pt}r{0pt}){6-7}
$F$ & Robust & Cluster & $\mathbf Y$-fixed & $\boldsymbol\epsilon$-fixed & $\mathbf Y$-fixed & $\boldsymbol\epsilon$-fixed \\
 & (1) & (2) & (3) & (4) & (5) & (6) \\
\midrule[0.45pt]
\multicolumn{7}{@{}l@{}}{\small\textbf{Panel A:} Positive homogeneous effect ($\beta=0.3$); independent errors} \\[2pt]
4 & 0.053 & 0.300 & 0.3871 & 0.8962 & -0.128 & 0.000 \\
8 & 0.051 & 0.136 & 0.2003 & 0.8553 & -0.113 & 0.000 \\
16 & 0.048 & 0.084 & 0.1103 & 0.9059 & -0.097 & 0.000 \\
100 & 0.050 & 0.054 & 0.0019 & 0.9956 & -- & 0.000 \\
\midrule[0.35pt]
\multicolumn{7}{@{}l@{}}{\small\textbf{Panel B:} Zero effect; independent errors} \\[2pt]
4 & 0.053 & 0.300 & 0.8509 & 0.8962 & 0.000 & 0.000 \\
8 & 0.051 & 0.136 & 0.7911 & 0.8553 & 0.000 & 0.000 \\
16 & 0.048 & 0.084 & 0.8716 & 0.9059 & 0.000 & 0.000 \\
100 & 0.050 & 0.054 & 0.9950 & 0.9956 & 0.000 & 0.000 \\
\midrule[0.35pt]
\multicolumn{7}{@{}l@{}}{\small\textbf{Panel C:} Positive homogeneous effect ($\beta=0.3$); within-group dependence} \\[2pt]
4 & 0.176 & 0.296 & 0.2977 & 0.6671 & -0.177 & 0.001 \\
8 & 0.176 & 0.139 & 0.0901 & 0.4066 & -0.173 & 0.001 \\
16 & 0.177 & 0.088 & 0.0198 & 0.2419 & -0.162 & 0.001 \\
100 & 0.178 & 0.058 & 0.0000 & 0.0072 & -- & -- \\
\midrule[0.35pt]
\multicolumn{7}{@{}l@{}}{\small\textbf{Panel D:} Zero effect; within-group dependence} \\[2pt]
4 & 0.176 & 0.296 & 0.5502 & 0.6671 & 0.001 & 0.001 \\
8 & 0.176 & 0.139 & 0.3000 & 0.4066 & 0.002 & 0.001 \\
16 & 0.177 & 0.088 & 0.1844 & 0.2419 & 0.000 & 0.001 \\
100 & 0.178 & 0.058 & 0.0055 & 0.0072 & -- & -- \\
\midrule[0.35pt]
\multicolumn{7}{@{}l@{}}{\small\textbf{Panel E:} Group-level heterogeneous effects; independent untreated outcomes} \\[2pt]
4 & 0.178 & 0.318 & 0.5766 & 0.6994 & 0.000 & 0.002 \\
8 & 0.175 & 0.145 & 0.3341 & 0.4421 & -0.001 & -0.001 \\
16 & 0.172 & 0.094 & 0.2149 & 0.2776 & 0.001 & 0.001 \\
100 & 0.173 & 0.056 & 0.0112 & 0.0135 & -0.002 & -0.001 \\
\bottomrule[0.45pt]
\end{tabular}
\par\medskip
\begin{minipage}{\linewidth}\footnotesize
Notes:  20,000 original samples per panel and group count, and $B=5,000$ independent balanced placebo assignments per sample. Tests use the standard normal critical value at nominal level $0.05$ and no degrees-of-freedom correction. Columns (1) and (2) report rejection probabilities for the true population ATE. Here $\widehat\Gamma_B$ denotes the fraction of the $B$ simulated regressions that reject zero using robust standard errors, with the $\mathbf Y$-fixed or $\boldsymbol\epsilon$-fixed construction indicated by the column. Approval requires  $\widehat\Gamma_B<0.10$. The same $\beta=0.3$ is used in panels A and C for every $F$. Panels C and D use $\omega=0.15$, while panel E uses $\operatorname{sd}(\eta_f)=\sqrt{2}(0.15)$ and $\omega=0$. Conditional bias is $\mathbb E[\widehat\beta-\beta\mid\widehat\Gamma_B<0.10]$. A dash in Columns (5) and (6) indicates that the estimated approval probability for the corresponding assessment is below $0.01$.
\end{minipage}
\end{table}

\clearpage
\begin{table}

\begin{footnotesize}

\begin{center}
\caption{{\bf Design-based Simulations for shift-share designs}} \label{Table_SS}
\begin{tabular}{lcccccccc}
\cline{1-9}
\cline{1-9}

&\multicolumn{2}{c}{China shock} & & \multicolumn{2}{c}{Exposure to robots} & & \multicolumn{2}{c}{Trade liberalization} \\ \cline{2-3} \cline{5-6} \cline{8-9}

& Main effects & Placebo && Main effects & Placebo & & Main effects & Placebo \\

 & (1) & (2) & & (3) & (4) & & (5) & (6) \\

\cline{1-9}

\multicolumn{9}{c}{Panel A: $\mathbf Y$-fixed design-based simulations for cluster-robust standard errors} \\
 \\

Cluster-robust & 0.412 & 0.237 & & 0.330 & 0.059 & & 0.313 & 0.001 \\
 \\

\multicolumn{9}{c}{Panel B: $\boldsymbol\epsilon$-fixed design-based simulations for cluster-robust standard errors} \\
\\
Cluster-robust & 0.317 & 0.226 & & 0.256 & 0.049 & & 0.012 & 0.000 \\
 \\

\multicolumn{9}{c}{Panel C: $\mathbf Y$-fixed  design-based simulations for shift-share-specific inference methods} \\
\\
\cite{Adao} & 0.076 & 0.098 & & 0.355 & 0.221 & & 0.547 & 0.511 \\
 \\
\cite{Adao} & 0.042 & 0.041 & & 0.296 & 0.107 & & 0.220 & 0.116 \\

(null imposed) \\
\\

\# of clusters & 48 & 48 & & 48 & 48 & & 91 & 91 \\

\# of observations & 722 & 722 & & 722 & 722 & & 411 & 411 \\

\# of sectors & 395 & 395 & & 19 & 19 & & 20 & 20 \\

\cline{1-9}
\cline{1-9}

\end{tabular}

\end{center}

\end{footnotesize}

\footnotesize{{Notes: the table compares assessments for the main and pre-treatment outcomes in three shift-share applications. Panels A and C hold observed outcomes fixed; Panel B removes the estimated treatment effect before fixing outcomes. Panels A and B assess cluster-robust inference with an HC3 correction. Panel C assesses the two procedures from \citet{Adao}, without and with the null imposed. All simulations use 30,000 iid standard normal shock vectors per specification and report rejection frequencies for nominal $5\%$ tests in unweighted OLS specifications. All tests use standard normal critical values. Columns 1 and 2 use the China-shock application of \citet{Autor}, with the placebo outcome measured from 1980 to 1990. Columns 3 and 4 use \citet{Acemoglu}, and columns 5 and 6 use \citet{Dix}. In the two U.S. applications, standard errors are clustered at the state level; in the Brazilian application, they are clustered at the mesoregion level.}}

\end{table}

\clearpage

\appendix
\setcounter{table}{0}
\renewcommand\thetable{A.\arabic{table}}
\setcounter{figure}{0}
\renewcommand\thefigure{A.\arabic{figure}}

\section{Appendix}
\onehalfspacing

\subsection{Proofs of the main results}
\label{app:main-proofs}

\subsubsection{Results from Section~\ref{Sec: DB in RCT}}
\label{app:rct-proof}

\begin{proof}[\normalfont\bfseries Proof of Proposition~\ref{prop:rct-assessment}]
Let $Q$ denote the common marginal distribution of $Y_i(0)$ and $Y_i(1)$, and let $P_T$ denote the assignment distribution. Fix any treatment allocation $t$ in the support of $P_T$. Conditional on $\mathbf T=t$, the observed outcome for individual $i$ is $Y_i(t_i)$. Since treatment assignment is independent of the potential outcomes and $Y_i(0)$ and $Y_i(1)$ have the same marginal distribution, $Y_i(t_i)$ has distribution $Q$ regardless of whether $t_i=0$ or $t_i=1$. Moreover, because the potential-outcome pairs are iid across individuals, the observed outcomes are independent across individuals conditional on $\mathbf T=t$. Therefore, the conditional distribution of $\mathbf Y$ is the same for every treatment allocation $t$. It follows that $\mathbf Y$ and $\mathbf T$ are independent.

Let $\varphi(y,t)$ be the rejection indicator for the chosen test of $H_0:\tau=0$. Then
\[
\Pr(R\mid\mathbf Y)
=
\sum_t \varphi(\mathbf Y,t)\Pr(\mathbf T=t\mid\mathbf Y)
=
\sum_t \varphi(\mathbf Y,t)P_T(t)
=
\Gamma_{\mathrm{rob}}^Y(\mathbf Y).
\]
The second equality holds because, under the assumptions in the proposition, the vector of observed outcomes and the treatment allocation are independent, and the last equality follows from the definition of $\Gamma_{\mathrm{rob}}^Y(\mathbf Y)$.

Taking expectations on both sides gives
\[
\Pr(R)=\mathbb E\!\left[\Gamma_{\mathrm{rob}}^Y(\mathbf Y)\right],
\]
which proves the proposition.

\end{proof}

\subsubsection{Results from Section~\ref{Sec: DB in SS}}
\label{Appendix_design_based}

\begin{proof}[\normalfont\bfseries Proof of Proposition~\ref{prop:original-ate}]

Because the group vectors are iid and independent of treatment, we may condition on any fixed balanced assignment. Let $\bar u_f=m^{-1}\sum_i u_{fi}$ and $\bar\zeta_f=m^{-1}\sum_i\zeta_{fi}$, and use subscripts $T$ and $C$ for averages over the treated and control groups. Then
\[
 \widehat\beta-\beta=\bar u_T-\bar u_C+\bar\eta_T+\bar\zeta_T.
\]
The variance of $\bar u_f$ is $[\sigma_u^2+(m-1)\rho]/m$, while the variance of $\bar\zeta_f$ is $\sigma_\zeta^2/m$. Independence across groups and independence of the treatment-effect components and baseline errors give
\begin{align*}
 \operatorname{Var}(\widehat\beta\mid\mathbf D)
 &=\frac2F\left[\frac{\sigma_u^2+(m-1)\rho}{m}+\sigma_\eta^2+\frac{\sigma_\zeta^2}{m}\right]
   +\frac2F\frac{\sigma_u^2+(m-1)\rho}{m}\\
 &=\frac4{Fm}\left[\sigma_u^2+(m-1)\rho+\frac{m\sigma_\eta^2}{2}+\frac{\sigma_\zeta^2}{2}\right].
\end{align*}

The conditional mean is $\beta$ and the variance is the same for every balanced assignment, so this is also the unconditional variance. Conditional on any balanced assignment, the treated and control samples consist of two disjoint collections of $F/2$ iid group-level observations. The central limit theorem applied to their averages, together with $N=Fm$ and fixed $m$, therefore yields the stated limiting distribution for $\sqrt N(\widehat\beta-\beta)$ conditional on the assignment. Because this conditional distribution is the same for every balanced assignment, the same convergence holds unconditionally over the treatment assignment.

For the robust variance estimator, the OLS residuals satisfy
\[
\widehat\epsilon_{fi}
=
\begin{cases}
\epsilon_{fi}-\bar\epsilon_T,& D_f=1,\\
\epsilon_{fi}-\bar\epsilon_C,& D_f=0,
\end{cases}
\]
where $\bar\epsilon_T$ and $\bar\epsilon_C$ denote the average errors in the treated and control groups. Therefore,
\[
\frac1N\sum_{fi}\widehat\epsilon_{fi}^2
=
\frac1N\sum_{fi}\epsilon_{fi}^2
-\frac12\bar\epsilon_T^2
-\frac12\bar\epsilon_C^2.
\]
By laws of large numbers across groups,
\[
\frac1N\sum_{fi}\epsilon_{fi}^2
\xrightarrow{p}
\sigma_u^2+\frac{\sigma_\eta^2+\sigma_\zeta^2}{2},
\qquad
\bar\epsilon_T,\bar\epsilon_C\xrightarrow{p}0.
\]
Hence,
\[
N\widehat V_{\mathrm{rob}}
=
\frac4N\sum_{fi}\widehat\epsilon_{fi}^2
\xrightarrow{p}
4\left[
\sigma_u^2+\frac{\sigma_\eta^2+\sigma_\zeta^2}{2}
\right].
\]

Similarly, letting $\bar\epsilon_f=m^{-1}\sum_i\epsilon_{fi}$, the average residual within group $f$ satisfies
\[
\overline{\widehat\epsilon}_f
=
\begin{cases}
\bar\epsilon_f-\bar\epsilon_T,& D_f=1,\\
\bar\epsilon_f-\bar\epsilon_C,& D_f=0.
\end{cases}
\]
It follows that
\[
\frac1F\sum_f\overline{\widehat\epsilon}_f^2
=
\frac1F\sum_f\bar\epsilon_f^2
-\frac12\bar\epsilon_T^2
-\frac12\bar\epsilon_C^2.
\]
Again by laws of large numbers across groups,
\[
\frac1F\sum_f\bar\epsilon_f^2
\xrightarrow{p}
\frac{\sigma_u^2+(m-1)\rho}{m}
+\frac{\sigma_\eta^2}{2}
+\frac{\sigma_\zeta^2}{2m},
\]
while $\bar\epsilon_T,\bar\epsilon_C\xrightarrow{p}0$. Therefore,
\[
N\widehat V_{\mathrm{cluster}}
=
\frac{4m}{F}\sum_f\overline{\widehat\epsilon}_f^2
\xrightarrow{p}
4\left[
\sigma_u^2+(m-1)\rho
+\frac{m\sigma_\eta^2}{2}
+\frac{\sigma_\zeta^2}{2}
\right].
\]
\end{proof}

\begin{proof}[\normalfont\bfseries Proofs of Propositions~\ref{prop:outcome-fixed} and~\ref{prop:residual-fixed}]

We first establish a result for a generic deterministic sequence of fixed outcome vectors. For each $F$, let $\mathbf v=\mathbf v_F$ contain $N=Fm$ observations, and define
\[
A_F(\mathbf v)=\frac1N\sum_{fi}(v_{fi}-\bar v)^2,
\qquad
B_F(\mathbf v)=\frac1F\sum_f(\bar v_f-\bar v)^2,
\]
where $\bar v_f$ and $\bar v$ are the group and overall means. Suppose that, along this deterministic sequence,
\[
A_F(\mathbf v)\longrightarrow a,
\qquad
B_F(\mathbf v)\longrightarrow b>0,
\]
and
\begin{equation}
H_F(\mathbf v)
\equiv
\frac{\max_f(\bar v_f-\bar v)^2}
{\sum_f(\bar v_f-\bar v)^2}
\longrightarrow 0.
\label{eq:no-dominant-group}
\end{equation}
We derive the asymptotic behavior of the design-based simulation under these conditions. Throughout this generic step, $\xrightarrow{p^*}$ and $\xrightarrow{d^*}$ denote convergence under the simulated assignment, with $\mathbf v$ held fixed.

Write $S_f^*=2D_f^*-1$. Since the simulated assignment is balanced, $\sum_f S_f^*=0$, and the coefficient in the simulated regression can be written as
\[
\widehat\gamma^*(\mathbf v)
=
\frac{2}{F}\sum_f S_f^*\bar v_f
=
\frac{2}{F}\sum_f S_f^*(\bar v_f-\bar v).
\]
Under balanced randomization, $\mathbb E^*[S_f^*]=0$, $\operatorname{Var}^*(S_f^*)=1$, and $\operatorname{Cov}^*(S_f^*,S_g^*)=-1/(F-1)$ for $f\neq g$. Therefore,
\[
V^*(\mathbf v)
=
\frac{4}{F^2}
\left[
\sum_f(\bar v_f-\bar v)^2
-\frac{1}{F-1}\sum_{f\neq g}
(\bar v_f-\bar v)(\bar v_g-\bar v)
\right].
\]
Since $\sum_f(\bar v_f-\bar v)=0$,
\[
\sum_{f\neq g}
(\bar v_f-\bar v)(\bar v_g-\bar v)
=
-\sum_f(\bar v_f-\bar v)^2,
\]
and hence
\begin{equation}
V^*(\mathbf v)
=
\frac{4}{F(F-1)}\sum_f(\bar v_f-\bar v)^2
=
\frac{4}{F-1}B_F(\mathbf v).
\label{eq:fixed-vector-variance}
\end{equation}
It follows that
\begin{equation}
N V^*(\mathbf v)\longrightarrow 4mb.
\label{eq:generic-true-variance}
\end{equation}

We next consider the robust variance estimator in the simulated regression. Let $\bar v_T^*$ and $\bar v_C^*$ denote the means of $\mathbf v$ in the simulated treatment and control arms. The residual for individual $i$ in group $f$ is
\[
\widehat r_{fi}^*(\mathbf v)
=
\begin{cases}
v_{fi}-\bar v_T^*,&D_f^*=1,\\
v_{fi}-\bar v_C^*,&D_f^*=0.
\end{cases}
\]
Since each simulated treatment arm contains $N/2$ observations,
\[
\widehat V_{\mathrm{rob}}^*(\mathbf v)
=
\frac{4}{N^2}\sum_{fi}\widehat r_{fi}^*(\mathbf v)^2.
\]
Moreover,
\[
\widehat\gamma^*(\mathbf v)=\bar v_T^*-\bar v_C^*,
\qquad
\bar v=\frac{\bar v_T^*+\bar v_C^*}{2}.
\]
The residual sum of squares satisfies
\begin{align*}
\sum_{fi}\widehat r_{fi}^*(\mathbf v)^2
&=
\sum_{fi}(v_{fi}-\bar v)^2
-\frac{N}{2}(\bar v_T^*-\bar v)^2
-\frac{N}{2}(\bar v_C^*-\bar v)^2\\
&=
\sum_{fi}(v_{fi}-\bar v)^2
-\frac{N}{4}\widehat\gamma^*(\mathbf v)^2.
\end{align*}
Therefore,
\begin{equation}
N\widehat V_{\mathrm{rob}}^*(\mathbf v)
=
4A_F(\mathbf v)-\widehat\gamma^*(\mathbf v)^2.
\label{eq:exact-sim-rob}
\end{equation}
Since $\mathbb E^*[\widehat\gamma^*(\mathbf v)]=0$ and $\mathbb E^*[\widehat\gamma^*(\mathbf v)^2]=V^*(\mathbf v)=O(F^{-1})$, we have $\widehat\gamma^*(\mathbf v)^2\xrightarrow{p^*}0$. Hence
\begin{equation}
N\widehat V_{\mathrm{rob}}^*(\mathbf v)\xrightarrow{p^*}4a,
\qquad
\frac{\widehat V_{\mathrm{rob}}^*(\mathbf v)}{V^*(\mathbf v)}
\xrightarrow{p^*}\frac{a}{mb}.
\label{eq:generic-variance-ratio}
\end{equation}

Condition~\eqref{eq:no-dominant-group} is the no-dominant-group condition for the finite-population central limit theorem. Since the simulated assignment selects $F/2$ of the $F$ group means without replacement,
\begin{equation}
\frac{\widehat\gamma^*(\mathbf v)}{\sqrt{V^*(\mathbf v)}}
\xrightarrow{d^*}N(0,1)
\label{eq:generic-fpclt}
\end{equation}
under the simulated assignment \citep{LiDing}. Combining \eqref{eq:generic-variance-ratio} and \eqref{eq:generic-fpclt} gives
\begin{equation}
\Gamma_{\mathrm{rob}}(\mathbf v)
\longrightarrow
2\left[1-\Phi\!\left(z_{1-\alpha/2}\sqrt{\frac{a}{mb}}\right)\right],
\label{eq:generic-gamma}
\end{equation}
where $\Gamma_{\mathrm{rob}}(\mathbf v)$ denotes the rejection probability of the design-based simulation that fixes $\mathbf v$.

We now apply this generic result to the random vectors generated by the original DGP. For the $\mathbf Y$-fixed simulations, laws of large numbers across groups imply
\begin{align*}
A_F(\mathbf Y)
&\xrightarrow{p}
\sigma_u^2+\frac{\sigma_\eta^2+\sigma_\zeta^2}{2}+\frac{\beta^2}{4}
\equiv a_Y,\\
B_F(\mathbf Y)
&\xrightarrow{p}
\frac{\sigma_u^2+(m-1)\rho}{m}
+\frac{\sigma_\eta^2}{2}
+\frac{\sigma_\zeta^2}{2m}
+\frac{\beta^2}{4}
\equiv b_Y.
\end{align*}
The finite fourth-moment assumption also implies
\[
\frac1F\max_f(\bar Y_f-\bar Y)^2\xrightarrow{p}0.
\]
Since $B_F(\mathbf Y)\xrightarrow{p}b_Y>0$, it follows that $H_F(\mathbf Y)\xrightarrow{p}0$.

The generic result above shows that for every deterministic sequence $\mathbf v_F$ satisfying
\[
A_F(\mathbf v_F)\to a,\qquad
B_F(\mathbf v_F)\to b>0,\qquad
H_F(\mathbf v_F)\to0,
\]
the corresponding rejection probability satisfies
\[
\Gamma_{\mathrm{rob}}(\mathbf v_F)
\to
G(a,b)
\equiv
2\left[1-\Phi\!\left(z_{1-\alpha/2}\sqrt{\frac{a}{mb}}\right)\right].
\]
Together with the convergence in probability of $(A_F(\mathbf Y),B_F(\mathbf Y),H_F(\mathbf Y))$, this implies
\[
\Gamma_{\mathrm{rob}}^Y(\mathbf Y)
\xrightarrow{p}G(a_Y,b_Y),
\]
which is \eqref{eq:gamma-outcome}. To see this, suppose otherwise. Then there exist $\varepsilon,\delta>0$ and a subsequence $F_k$ such that $\Pr(|\Gamma_{\mathrm{rob}}^Y(\mathbf Y)-G(a_Y,b_Y)|>\varepsilon)>\delta$. Passing to a further subsequence if necessary, convergence in probability of $(A_F,B_F,H_F)$ implies that the probability that this vector is within $1/k$ of $(a_Y,b_Y,0)$ exceeds $1-\delta/2$. The intersection of these two events therefore has positive probability, so for each $k$ we can select a deterministic vector $\mathbf v_{F_k}$ for which $(A_{F_k},B_{F_k},H_{F_k})\to(a_Y,b_Y,0)$ while $|\Gamma_{\mathrm{rob}}(\mathbf v_{F_k})-G(a_Y,b_Y)|>\varepsilon$, contradicting the deterministic-sequence result.

For the $\boldsymbol\epsilon$-fixed simulations, the identities for the OLS residuals used in the proof of Proposition~\ref{prop:original-ate} give
\begin{align*}
A_F(\widehat{\boldsymbol\epsilon})
&\xrightarrow{p}
\sigma_u^2+\frac{\sigma_\eta^2+\sigma_\zeta^2}{2}
\equiv a_\epsilon,\\
B_F(\widehat{\boldsymbol\epsilon})
&\xrightarrow{p}
\frac{\sigma_u^2+(m-1)\rho}{m}
+\frac{\sigma_\eta^2}{2}
+\frac{\sigma_\zeta^2}{2m}
\equiv b_\epsilon.
\end{align*}
Finite fourth moments similarly imply
\[
\frac1F\max_f\overline{\widehat\epsilon}_f^{\,2}\xrightarrow{p}0,
\]
so $H_F(\widehat{\boldsymbol\epsilon})\xrightarrow{p}0$. Applying the same argument with $a=a_\epsilon$ and $b=b_\epsilon$ gives
\[
\Gamma_{\mathrm{rob}}^\epsilon(\mathbf Y,\mathbf D)
\xrightarrow{p}G(a_\epsilon,b_\epsilon),
\]
which is \eqref{eq:gamma-residual}. This establishes the conclusions of Proposition~\ref{prop:residual-fixed}.

The expression in \eqref{eq:gamma-residual} is also the limiting rejection probability of the robust test of the true population ATE in the original experiment. Proposition~\ref{prop:original-ate} implies
\[
\frac{\widehat\beta-\beta}{\sqrt{\widehat V_{\mathrm{rob}}}}
\xrightarrow{d}
N\!\left(
0,
\frac{
\sigma_u^2+(m-1)\rho+m\sigma_\eta^2/2+\sigma_\zeta^2/2
}{
\sigma_u^2+(\sigma_\eta^2+\sigma_\zeta^2)/2
}
\right),
\]
which yields the same rejection-probability limit.

Finally, for $\rho\geq0$ and $m>1$,
\[
mb_Y-a_Y
=(m-1)\left(\rho+\frac{\sigma_\eta^2}{2}+\frac{\beta^2}{4}\right),
\]
while
\[
mb_\epsilon-a_\epsilon
=(m-1)\left(\rho+\frac{\sigma_\eta^2}{2}\right).
\]
The characterizations of when the limiting rejection probabilities equal or exceed $\alpha$ follow immediately from \eqref{eq:generic-gamma}.
\end{proof}

\label{app:finite-group-assessment}


\begin{proof}[\normalfont\bfseries Proof of Proposition~\ref{prop:finite-group-assessment}]
Let $Q_f$ denote the common distribution of $\mathbf Y_f(0)$ and $\mathbf Y_f(1)$ for group $f$, and let $P_D$ denote the assignment distribution. Fix any treatment allocation $d$ in the support of $P_D$. Conditional on $\mathbf D=d$, the observed outcome vector for group $f$ is $\mathbf Y_f(d_f)$. Since $\mathbf Y_f(0)$ and $\mathbf Y_f(1)$ have the same distribution, $\mathbf Y_f(d_f)$ has distribution $Q_f$ regardless of whether $d_f=0$ or $d_f=1$. Moreover, because the potential-outcome pairs are independent across groups and independent of treatment assignment, the conditional distribution of $\mathbf Y$ is the same for every treatment allocation $d$. It follows that $\mathbf Y$ and $\mathbf D$ are independent.

Let $\varphi_j(y,d)$ be the rejection indicator for method $j$. Then
\[
\Pr(R_j\mid\mathbf Y)
=
\sum_d\varphi_j(\mathbf Y,d)\Pr(\mathbf D=d\mid\mathbf Y)
=
\sum_d\varphi_j(\mathbf Y,d)P_D(d)
=
\Gamma_j^Y(\mathbf Y).
\]
The second equality holds because the vector of observed outcomes and the treatment allocation are independent, and the last equality follows from the definition of $\Gamma_j^Y(\mathbf Y)$.

Taking expectations on both sides gives
\[
\Pr(R_j)=\mathbb E\!\left[\Gamma_j^Y(\mathbf Y)\right],
\]
which proves the proposition.
\end{proof}

\subsubsection{Results from Section~\ref{sec:selection}}
\label{app:selection-proofs}

For a fixed cutoff $\kappa$, let
\[
A_\kappa
=
\left\{
\Gamma_{\mathrm{rob}}^\epsilon(\mathbf Y,\mathbf D)<\kappa
\right\},
\]
and suppose $\Pr(A_\kappa)>0$. Since the assessment is based on a two-sided test, $\Gamma_{\mathrm{rob}}^\epsilon(\mathbf Y,\mathbf D)$ is unchanged when $\widehat{\boldsymbol\epsilon}$ is replaced by $-\widehat{\boldsymbol\epsilon}$.

\begin{proposition}[Absence of bias after $\boldsymbol\epsilon$-fixed selection]\label{prop:symmetry}

Suppose, in addition to the group-assignment model, that
\[
(u_f,\eta_f,\zeta_f)\overset d=(-u_f,-\eta_f,-\zeta_f).
\]
Then
\[
\mathbb E[\widehat\beta\mid A_\kappa]=\beta.
\]

\end{proposition}

\begin{proof}[\normalfont\bfseries Proof of Proposition~\ref{prop:symmetry}]

Conditional on $\mathbf D$, the regression error is
\[
\epsilon_{fi}=u_{fi}+\eta_fD_f+\zeta_{fi}D_f.
\]
The coefficient error $\widehat\beta-\beta$ is the treated-minus-control mean of $\epsilon_{fi}$, while $\widehat{\boldsymbol\epsilon}$ is obtained by subtracting the corresponding treatment-arm means. Under the symmetry assumption,
\[
(\widehat\beta-\beta,\widehat{\boldsymbol\epsilon})
\overset d=
(-[\widehat\beta-\beta],-\widehat{\boldsymbol\epsilon})
\qquad\text{conditional on }\mathbf D.
\]
Moreover, the event $A_\kappa$ is unchanged when $\widehat{\boldsymbol\epsilon}$ is replaced by $-\widehat{\boldsymbol\epsilon}$. Hence
\[
(\widehat\beta-\beta)\mathbf1\{A_\kappa\}
\overset d=
-(\widehat\beta-\beta)\mathbf1\{A_\kappa\}
\qquad\text{conditional on }\mathbf D.
\]
Therefore,
\[
\mathbb E\!\left[
(\widehat\beta-\beta)\mathbf1\{A_\kappa\}
\,\middle|\,\mathbf D
\right]
=0.
\]Taking expectations over $\mathbf D$ gives
\[
\mathbb E\!\left[
(\widehat\beta-\beta)\mathbf1\{A_\kappa\}
\right]
=0.
\]
Since $\Pr(A_\kappa)>0$, it follows that
\[
\mathbb E[\widehat\beta\mid A_\kappa]=\beta.
\]
\end{proof}

\begin{proposition}[Preservation of the coefficient distribution under Gaussian errors]\label{prop:gaussian-independence}

Suppose $u_f$ is Gaussian with covariance
\[
(\sigma_u^2-\rho)I_m+\rho\mathbf1\mathbf1',
\]
$\eta_f\sim N(0,\sigma_\eta^2)$ and $\zeta_{fi}\sim N(0,\sigma_\zeta^2)$, with the independence assumptions of the group model. Then, for every $c\in\mathbb R$,
\[
\Pr(\widehat\beta\leq c\mid A_\kappa)
=
\Pr(\widehat\beta\leq c).
\]
Thus conditioning on a favorable $\boldsymbol\epsilon$-fixed assessment preserves the entire distribution of $\widehat\beta$.

\end{proposition}

\begin{proof}[\normalfont\bfseries Proof of Proposition~\ref{prop:gaussian-independence}]

Conditional on $\mathbf D$, the full error vector $\boldsymbol\epsilon$ is multivariate Gaussian. Moreover,
\[
\widehat\beta-\beta=\bar\epsilon_T-\bar\epsilon_C,
\]
and, for an observation in arm $a\in\{T,C\}$,
\[
\widehat\epsilon_{fi}=\epsilon_{fi}-\bar\epsilon_a.
\]
Therefore, $(\widehat\beta-\beta,\widehat{\boldsymbol\epsilon})$ is jointly Gaussian conditional on $\mathbf D$.

For a control observation,
\[
\operatorname{Cov}(\bar\epsilon_C,\epsilon_{fi}\mid\mathbf D)
=
\frac{2}{Fm}\left[\sigma_u^2+(m-1)\rho\right]
=
\operatorname{Var}(\bar\epsilon_C\mid\mathbf D),
\]
while for a treated observation,
\[
\operatorname{Cov}(\bar\epsilon_T,\epsilon_{fi}\mid\mathbf D)
=
\frac{2}{Fm}\left[
\sigma_u^2+(m-1)\rho
+m\sigma_\eta^2+\sigma_\zeta^2
\right]
=
\operatorname{Var}(\bar\epsilon_T\mid\mathbf D).
\]
Since errors are independent across treatment arms, it follows that
\[
\operatorname{Cov}
\left(
\widehat\beta-\beta,\widehat\epsilon_{fi}
\,\middle|\,\mathbf D
\right)
=0
\]
for every observation. Joint Gaussianity therefore implies
\[
\widehat\beta-\beta
\perp
\widehat{\boldsymbol\epsilon}
\mid\mathbf D.
\]

Conditional on $\mathbf D$, the event $A_\kappa$ is a function only of $\widehat{\boldsymbol\epsilon}$. Hence, for every $c$,
\[
\Pr(\widehat\beta\leq c\mid A_\kappa,\mathbf D)
=
\Pr(\widehat\beta\leq c\mid\mathbf D).
\]
Moreover, $A_\kappa$ is independent of $\mathbf D$. To see this, note that for any balanced assignment there are $F/2$ treated and $F/2$ control groups. Since the group vectors are iid, the conditional distribution of the unrestricted residual vector is the same, up to a relabeling of groups, for every balanced assignment. The $\boldsymbol\epsilon$-fixed assessment is invariant to such relabeling, so $\Pr(A_\kappa\mid\mathbf D=d)$ is the same for every $d$ in the support of the design. Therefore,
\[
\Pr(\mathbf D=d\mid A_\kappa)=\Pr(\mathbf D=d).
\]
Averaging the first equality over $\mathbf D$ conditional on $A_\kappa$ then gives
\begin{align*}
\Pr(\widehat\beta\leq c\mid A_\kappa)
&=
\sum_d
\Pr(\widehat\beta\leq c\mid\mathbf D=d)
\Pr(\mathbf D=d\mid A_\kappa)\\
&=
\sum_d
\Pr(\widehat\beta\leq c\mid\mathbf D=d)
\Pr(\mathbf D=d)
=
\Pr(\widehat\beta\leq c),
\end{align*}
which proves the result.
\end{proof}
\subsection{An individual-randomization example}
\label{app:individual-coverage}

Consider a completely randomized experiment with ten individuals, five of whom receive treatment. We use a realization of ten iid $N(0,1)$ draws, selected to illustrate a pronounced nonmonotonic response of the average assessment $\Gamma_{\mathrm{rob}}^Y$ to changes in $\tau$.\footnote{The selected values, rounded to six decimal places, are $-0.840855$, $1.384359$, $-1.255492$, $0.070143$, $1.711441$, $-0.602908$, $-0.472166$, $-0.635371$, $-0.285774$, and $0.138108$.} We set $Y_i(0)=e_i$ and $Y_i(1)=e_i+\tau$ and hold $\mathbf e$ fixed throughout the exercise. We estimate the difference in means by OLS with an intercept and use robust confidence intervals with the standard normal critical value $z_{.975}$. Both the original and simulated regressions use the robust variance estimator without a degrees-of-freedom correction, $\widehat V_{\mathrm{rob}}=4\sum_i\widehat\epsilon_i^2/N^2$.

The true coverage of the interval for the SATE, $\tau$, is constant in $\tau$. For every original allocation $\mathbf T$,
\begin{align*}
 \widehat\tau(\mathbf e+\tau\mathbf T,\mathbf T)-\tau
 &=\widehat\tau(\mathbf e,\mathbf T),\\
 \widehat V_{\mathrm{rob}}(\mathbf e+\tau\mathbf T,\mathbf T)
 &=\widehat V_{\mathrm{rob}}(\mathbf e,\mathbf T).
\end{align*}

Therefore, adding the same effect to all treated outcomes changes the coefficient by $\tau$ without changing the regression residuals. In contrast, the design-based simulation fixes $\mathbf Y=\mathbf e+\tau\mathbf T$ as both potential outcomes and evaluates coverage of zero, the SATE in its artificial experiment.

Figure~\ref{fig:individual-coverage-balanced} enumerates all $\binom{10}{5}=252$ original allocations and all 252 placebo allocations for each realized outcome. The dashed line is true coverage, and the dots report the average of $1-\Gamma_{\mathrm{rob}}^Y(\mathbf e+\tau\mathbf T)$ over original allocations, with $\mathbf e$ fixed. True coverage is $84.92\%$ for every $\tau$, whereas average simulation coverage varies with $\tau$. It equals true coverage at $\tau=0$, rises to $88.45\%$ at $\tau=2.4$, and then falls below true coverage, reaching $79.54\%$ at $\tau=15$. Thus the average assessment first decreases and then increases as $\tau$ changes, even though the true rejection probability remains constant.

\subsection{A stylized DiD model with AR(1) errors}
\label{app:did}

We consider a four-period version of a standard DiD design to simplify the exposition. There are two pre-treatment periods, $t=1,2$, and two post-treatment periods, $t=3,4$. Let $P_t=\mathbf1\{t\geq3\}$, and suppose exactly half of the $F$ groups are assigned treatment, independently of potential outcomes, with all balanced assignments equally likely. 

We assume homogeneous treatment effects so the outcome equation is
\begin{equation}
 Y_{ft}=a_f+\lambda_t+\beta D_fP_t+\epsilon_{ft}.
 \label{eq:did-model}
\end{equation}
The group and period effects are fixed. Error histories are iid across groups and follow a stationary mean-zero AR(1) process,
\[
 \epsilon_{ft}=\gamma\epsilon_{f,t-1}+\nu_{ft},
 \qquad 0\leq\gamma<1,
 \qquad
 \operatorname{Cov}(\epsilon_{ft},\epsilon_{fs})
 =\sigma_\epsilon^2\gamma^{|t-s|},
\]
with independent innovations, marginal variance $\sigma_\epsilon^2>0$, and finite fourth moments. We consider $F\to\infty$.

Let $\widehat\beta$ be the coefficient on $D_fP_t$ from the regression with group and period fixed effects, and let $\widehat\epsilon_{ft}$ be its unrestricted residuals. For the robust variance estimator, we use the full-regression degrees-of-freedom correction. With four periods,
\begin{equation}
 \widehat V_{\mathrm{rob}}
 =\frac{4}{F(3F-4)}\sum_{f,t}\widehat\epsilon_{ft}^2.
 \label{eq:did-variance}
\end{equation}

For each simulated assignment $\mathbf D^*$, we retain the same post-treatment periods and fixed effects and reestimate the full regression. The $\mathbf Y$-fixed simulations use the observed panel, whereas the $\boldsymbol\epsilon$-fixed simulations use the unrestricted residual vector $\widehat{\boldsymbol\epsilon}$.

\begin{proposition}[Design-based simulations in the DiD model]
\label{prop:did-results}\label{prop:did-original}\label{prop:did-assessments}\label{prop:did-outcome}\label{prop:did-residual}
Under the assumptions above,
\begin{align}
 \Gamma_{\mathrm{rob}}^Y(\mathbf Y)
 &\xrightarrow{p}
 2\left[1-\Phi\!\left(
 z_{1-\alpha/2}
 \sqrt{
 \frac{
 2\sigma_\epsilon^2(1-\gamma)(\gamma^2+3\gamma+6)+\beta^2
 }{
 3\left[2\sigma_\epsilon^2(1-\gamma^2)(\gamma+2)+\beta^2\right]
 }
 }
 \right)\right],
 \label{eq:did-gamma-outcome}\\
 \Gamma_{\mathrm{rob}}^\epsilon(\mathbf Y,\mathbf D)
 &\xrightarrow{p}
 2\left[1-\Phi\!\left(
 z_{1-\alpha/2}
 \sqrt{
 \frac{\gamma^2+3\gamma+6}
 {3(\gamma+1)(\gamma+2)}
 }
 \right)\right].
 \label{eq:did-gamma-residual}
\end{align}
The $\boldsymbol\epsilon$-fixed limit is also the limiting rejection probability of the robust test of the true treatment effect in the original experiment. It equals $\alpha$ when $\gamma=0$ and exceeds $\alpha$ when $\gamma>0$. The $\mathbf Y$-fixed limit equals the $\boldsymbol\epsilon$-fixed limit when $\beta=0$ and is strictly larger when $\beta\neq0$.
\end{proposition}

Thus the four-period DiD model yields the same qualitative conclusions as Propositions~\ref{prop:outcome-fixed} and~\ref{prop:residual-fixed}.

\begin{proof}[\normalfont\bfseries Proof of Proposition~\ref{prop:did-results}]
Let $s=(-1,-1,1,1)'$ and define the pre-post contrast
\[
 q_f(\mathbf v)=\frac12\sum_{t=1}^4 s_t v_{ft}.
\]
After removing group and period fixed effects, the DiD estimator satisfies
\[
 \widehat\beta-\beta
 =\frac2F\sum_f(2D_f-1)q_f(\boldsymbol\epsilon).
\]
For the AR(1) errors,
\[
 \operatorname{Var}\!\left(q_f(\boldsymbol\epsilon)\right)
 =\frac{\sigma_\epsilon^2}{2}(1-\gamma^2)(\gamma+2).
\]
The central limit theorem across groups therefore gives
\[
 \sqrt F(\widehat\beta-\beta)
 \xrightarrow{d}
 N\!\left(0,2\sigma_\epsilon^2(1-\gamma^2)(\gamma+2)\right).
\]
Moreover, estimation of the common period effects and of $\beta$ contributes $o_p(F)$ to the residual sum of squares, so laws of large numbers across groups give
\[
 \frac1F\sum_{f,t}\widehat\epsilon_{ft}^2
 \xrightarrow{p}
 \frac{\sigma_\epsilon^2}{2}(1-\gamma)(\gamma^2+3\gamma+6).
\]
Using \eqref{eq:did-variance},
\[
 F\widehat V_{\mathrm{rob}}
 \xrightarrow{p}
 \frac{2\sigma_\epsilon^2}{3}(1-\gamma)(\gamma^2+3\gamma+6).
\]
Combining these two limits shows that the limiting rejection probability of the original robust test is the expression in \eqref{eq:did-gamma-residual}.

We next consider the design-based simulations. For any fixed panel $\mathbf v$,
\[
 \widehat\gamma^*(\mathbf v)
 =\frac2F\sum_f(2D_f^*-1)q_f(\mathbf v),
 \qquad
 V^*(\mathbf v)
 =\frac4{F(F-1)}\sum_f\left(q_f(\mathbf v)-\bar q(\mathbf v)\right)^2.
\]
Let $\ddot v_{ft}$ denote the outcome after removing group and period means. For the observed panel, laws of large numbers imply
\begin{align*}
 \frac1F\sum_f\left(q_f(\mathbf Y)-\bar q(\mathbf Y)\right)^2
 &\xrightarrow{p}
 \frac{\sigma_\epsilon^2}{2}(1-\gamma^2)(\gamma+2)+\frac{\beta^2}{4},\\
 \frac1F\sum_{f,t}\ddot Y_{ft}^{\,2}
 &\xrightarrow{p}
 \frac{\sigma_\epsilon^2}{2}(1-\gamma)(\gamma^2+3\gamma+6)+\frac{\beta^2}{4}.
\end{align*}
Finite fourth moments imply the corresponding no-dominant-group condition for the group contrasts. Together with the two moment limits above, the argument in the proof of Proposition~\ref{prop:outcome-fixed} gives \eqref{eq:did-gamma-outcome}. For the $\boldsymbol\epsilon$-fixed simulations, subtracting the estimated treatment effect removes the $\beta^2/4$ terms from both limits, giving \eqref{eq:did-gamma-residual}.

It remains to compare the limits. The fraction inside the square root in \eqref{eq:did-gamma-residual} equals one at $\gamma=0$ and is below one for $0<\gamma<1$, since
\[
 3(\gamma+1)(\gamma+2)-(\gamma^2+3\gamma+6)
 =2\gamma(\gamma+3)>0.
\]
It is also strictly larger than $1/3$, since
\[
 \frac{\gamma^2+3\gamma+6}{3(\gamma+1)(\gamma+2)}-\frac13
 =\frac{4}{3(\gamma+1)(\gamma+2)}>0.
\]
The fraction inside the square root in \eqref{eq:did-gamma-outcome} can be written as a weighted average of the fraction in \eqref{eq:did-gamma-residual} and $1/3$, with positive weight on $1/3$ whenever $\beta\neq0$. It therefore equals the $\boldsymbol\epsilon$-fixed fraction when $\beta=0$ and is strictly smaller when $\beta\neq0$, which proves the stated rejection-probability comparisons.
\end{proof}

\subsection{Details on the Monte Carlo simulations}
\label{Appendix_permutation_shocks}

The simulations in Table~\ref{tab:dbs-unified} use $F=4,8,16,100$ groups with $m=50$ individuals each. Half of the groups receive treatment, independently of potential outcomes. The DGP is
\[
 Y_{fi}(0)=\omega\xi_f+e_{fi},
 \qquad
 Y_{fi}(1)=Y_{fi}(0)+\beta+\eta_f,
\]
where the $e_{fi}$ and $\xi_f$ are independent standard normal variables. The group-level treatment-effect components $\eta_f$ are independent normal variables with mean zero and variance $\sigma_\eta^2$, independent of the other components. Thus $\sigma_u^2=1+\omega^2$, $\rho=\omega^2$, and $\sigma_\zeta^2=0$ in the notation of Section~\ref{sec:group-model}.

Panels A--D have homogeneous effects. Panels A and C set $\beta=0.3$, while B and D set $\beta=0$. Panels C and D set $\omega=0.15$, while A and B set $\omega=0$. Panel E sets $\beta=\omega=0$ and $\sigma_\eta=\sqrt{2}(0.15)$. All parameters are held fixed when $F$ changes. The population ATE is $\beta$ in every panel. In C, D, and E, the limiting robust rejection probability is approximately $0.174$.

For each configuration, we draw 20,000 original samples. The original robust and cluster-robust tests use the variance estimators in \eqref{eq:variance-estimators}, without degrees-of-freedom corrections, and the standard normal critical value $z_{.975}$. Each test is centered at the true population ATE, including $\beta=0.3$ in Panels A and C. For every original sample, we draw $B=5,000$ independent balanced placebo assignments with replacement, using the same draws for the $\mathbf Y$-fixed and $\boldsymbol\epsilon$-fixed simulations. These simulated regressions use robust standard errors without corrections and test zero, which is true in their artificial DGPs. Approval requires strictly $\widehat\Gamma_B<0.10$, or fewer than 500 rejections among the 5,000 simulated regressions.

The probability of approval is estimated by the fraction of original samples that pass this rule. Conditional bias is the mean of $\widehat\beta-\beta$ among approved samples, calculated separately for each construction. We omit the displayed bias when its corresponding approval frequency is below $0.01$. 

\subsection{Probability of flagging spatial correlation}
\label{Appendix_SS_applications}

We consider the probability that design-based simulations detect spatial-correlation problems using the shares from the three applications in Section~\ref{Sec: SS illustrations}. For each application, we construct an outcome and regressor as
\[
 Y_i^*=Z_i+\gamma\sum_{f=1}^Fw_{if}\mathcal X_f^*,
 \qquad
 X_i=\sum_{f=1}^Fw_{if}\mathcal X_f,
\]
where $Z_i$, $\mathcal X_f^*$, and $\mathcal X_f$ are mutually independent standard normal variables. The shocks $\mathcal X_f^*$ are unobserved outcome components, distinct from the shocks used to construct the regressor. Since they use the same shares, they can generate correlation between errors in regions with similar exposure. Its magnitude increases with $|\gamma|$. The outcome and regressor are independent, so the true coefficient is zero.

We retain the observed sample, shares, covariates, and geographic clusters from the main specification of each application in Table~\ref{Table_SS}, and use unweighted OLS. Both the original tests and the placebo tests use the same cluster-HC3 variance estimator and standard normal critical values at the nominal $5\%$ level as Table~\ref{Table_SS}. 

For each application and value of $\gamma$, we draw 5,000 original samples. In each sample, we estimate the shift-share regression, test zero, and construct both $\mathbf Y$-fixed and $\boldsymbol\epsilon$-fixed assessments using $B=1,000$ independent placebo shock vectors. These draws are independent of the original outcome and regressor shocks. Common random numbers are used across values of $\gamma$ and across the two assessment constructions. For each DGP, we estimate the true rejection probability and the probability that the assessment is at least $0.10$. The flagging rule $\widehat\Gamma_B\geq0.10$ is the complement of the approval rule $\widehat\Gamma_B<0.10$ used in Table~\ref{tab:dbs-unified}.

For each application, we consider different values of $\gamma$, which generate different combinations of true rejection probabilities and probabilities that the assessment flags a problem. Figure~\ref{Fig_SS_flagging} plots the flagging probability against the corresponding true rejection probability for each value of $\gamma$. The results are discussed in Section~\ref{sec:assessing-cluster-inference}.

\clearpage
\begin{figure}[p]
\centering
\caption{True coverage and coverage in $\mathbf Y$-fixed design-based simulations}
\label{fig:individual-coverage-balanced}
\includegraphics[width=\textwidth]{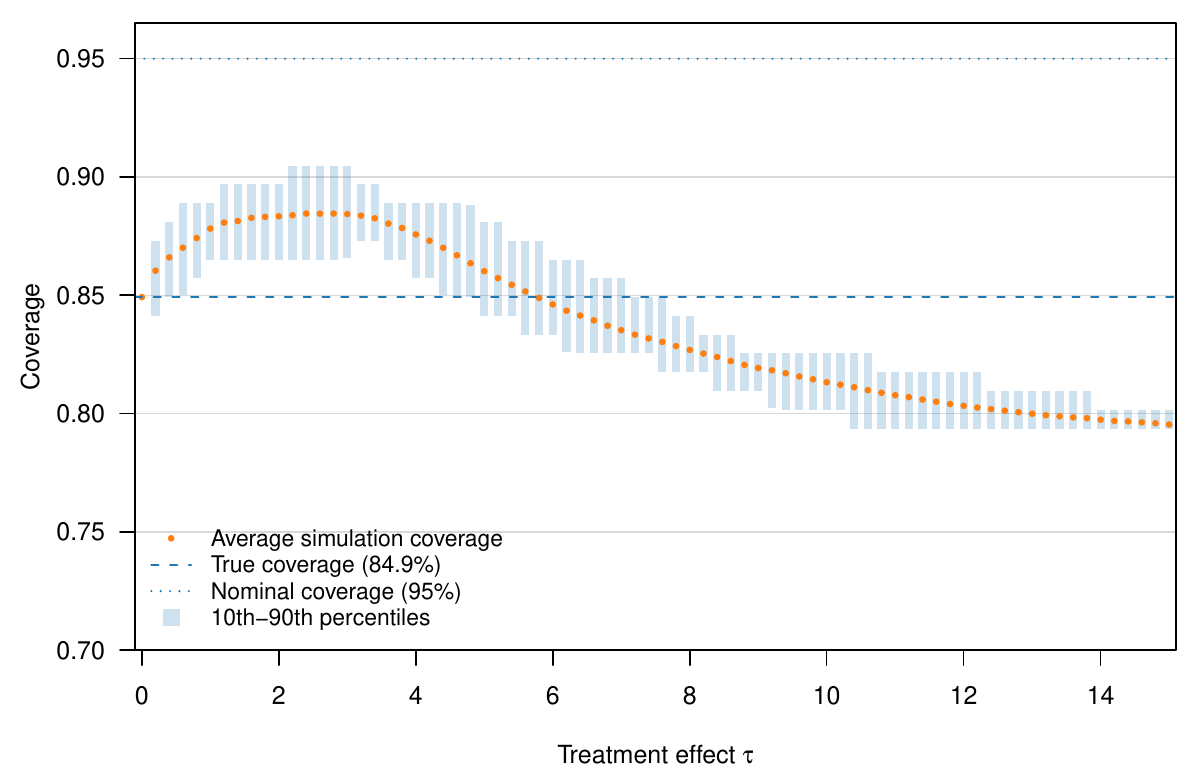}
\par\medskip
\begin{minipage}{\textwidth}
\small Notes: the figure uses the individual-randomization experiment described in Appendix~\ref{app:individual-coverage}. For each treatment effect $\tau$, the dots report average coverage in the design-based simulations across original treatment allocations, and the shaded bars span the corresponding 10th to 90th percentiles. The dashed line reports true coverage, and the dotted line indicates the nominal coverage of $95\%$.
\end{minipage}
\end{figure}

\clearpage
\begin{figure}

\begin{center}

\caption{{\bf Probability of detecting spatial correlation problems}} \label{Fig_SS_flagging}

\includegraphics[width=\textwidth]{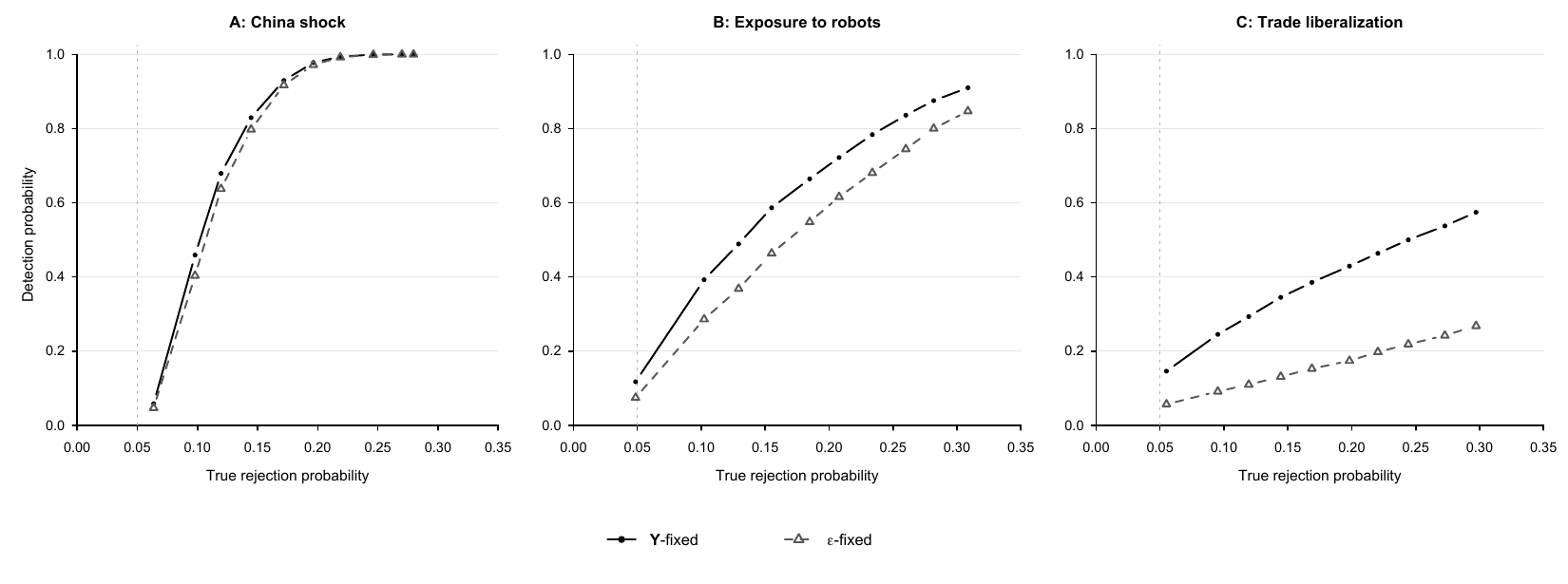}

\end{center}

\small{{Notes: see Appendix~\ref{Appendix_SS_applications} for simulation details. Each panel corresponds to one application, and each point to a value of $\gamma$. The horizontal axis reports the true rejection probability, while the vertical axis reports the probability that the design-based assessment is at least $0.10$. Solid lines with circles show the $\mathbf Y$-fixed assessment, while dashed lines with triangles show the $\boldsymbol\epsilon$-fixed assessment. }}

\end{figure}
\clearpage

\end{document}